%% file: phaseless.tex
\documentclass[showkeys,aps,nofootinbib,onecolumn,floatfix]{amsart}

\usepackage{amsmath,amsthm,amsfonts,amssymb,times,bbm,graphicx}
\usepackage{verbatim}
\usepackage{tikz}
\usepackage{scrextend}
\usepackage[norelsize]{algorithm2e}
\usetikzlibrary{arrows,calc,positioning}
\usetikzlibrary{circuits.logic.IEC}
\usepackage{xcolor}
\usetikzlibrary{shapes}
\usetikzlibrary{decorations.markings}
\input{tikzstyledefs.tex}
\usepackage{afterpage}

\usepackage{hyperref}

\newtheorem{theorem}{Theorem}
\newtheorem{proposition}[theorem]{Proposition}
\newtheorem{lemma}[theorem]{Lemma}

\newtheorem{definition}[theorem]{Definition}

\def\RR{\mathbbm{R}}
\def\CC{\mathbbm{C}}

\def\EE{\mathbbm{E}}
\def\AA{\mathcal{A}}
\def\II{\mathcal{I}}

\def\MM{\mathcal{M}}
\def\op{\mathrm{op}}

\def\Id{\mathbbm{1}}

\def\RR{\mathcal{R}}
\def\PP{\mathcal{P}}
\def\II{\mathcal{I}}

\DeclareMathOperator{\Sym}{Sym}

\DeclareMathOperator{\tr}{tr}
\DeclareMathOperator{\Tr}{Tr}

\DeclareMathOperator{\range}{range}

\makeatletter
\g@addto@macro{\endabstract}{\@setabstract}
\newcommand{\authorfootnotes}{\renewcommand\thefootnote{\@fnsymbol\c@footnote}}
\makeatother

\begin{document}

\begin{center}
  \LARGE 
  A Partial Derandomization of PhaseLift using Spherical Designs \par \bigskip

\normalsize
\authorfootnotes
D.\ Gross\textsuperscript{1,2}, 
F.\ Krahmer\textsuperscript{3},
R.\ Kueng\footnote{Corresponding author: richard.kueng@physik.uni-freiburg.de}\textsuperscript{1} \par \bigskip

\small
\textsuperscript{1}Institute for Physics, University of Freiburg, Rheinstra{\ss}e 10, 79104 Freiburg, Germany \par
\textsuperscript{2}Freiburg Center for Data Analysis Modeling, Eckerstr.~1, 79104 Freiburg, Germany \par
\textsuperscript{3}Institute for Numerical and Applied Mathematics, University of G{\"o}ttingen, Lotzestra\ss e 16-18, 37083 G{\"o}ttingen, Germany\par \bigskip

  \today
\end{center}

\begin{abstract} 
The problem of retrieving phase information from amplitude
measurements alone has appeared in many scientific disciplines over the
last century. 
\emph{PhaseLift} is a recently introduced algorithm for phase recovery that
is computationally tractable, numerically stable, and comes with
rigorous performance guarantees.
PhaseLift is optimal in the sense that the number of amplitude
measurements required for phase reconstruction scales linearly with
the dimension of the signal. However, it specifically demands
Gaussian random measurement vectors --- a limitation that restricts
practical utility and obscures the 
specific properties of measurement
ensembles that enable phase retrieval.
Here we present a partial derandomization of PhaseLift that only
requires sampling from certain polynomial size vector configurations, called
\emph{$t$-designs}. Such configurations have been studied in algebraic
combinatorics, coding theory, and quantum information.
We prove reconstruction guarantees for a number of measurements that
depends on the degree $t$ of the design. If the degree is allowed 
to grow logarithmically with the dimension, the bounds become 
tight up to polylog-factors.
Beyond the specific case of PhaseLift, this work highlights the utility of
spherical designs for the derandomization of data
recovery schemes.



\end{abstract}

\section{Introduction}

In this work we are interested in the problem of recovering a complex
signal (vector) $x \in \CC^d$ from an \emph{intensity} measurement
$y_0 = \| x \|_{\ell_2}^2$ and \emph{amplitude} measurements
\begin{equation*} y_i = | \langle a_i, x \rangle |^2 \quad i=1,\ldots,
m,	\label{eq:measurements} \end{equation*} where $a_1, \ldots, a_m
\in \CC^d$ are sampling vectors.  Problems of this type are abundant
in many different areas of science, where capturing phase information
is hard or even infeasible, but obtaining amplitudes is comparatively
easy.  Prominent examples for this case occur in X-ray
cristallography, astronomy and diffraction imaging -- see for example
\cite{M90}.
This inverse problem is called \emph{phase
retrieval} and has attracted considerable interest over the last
decades. 

It is by no means clear how many such amplitude measurements are necessary to 
allow for recovery. Thus from the very beginning, there have been a number of works regarding 
injectivity conditions for this problem in the context of the specific applications \cite{BS79}.

More recently this question has been studied in more abstract terms, asking
for the minimal number of amplitude measurements of the form (\ref{eq:measurements}) -- without imposing structural assumptions on the $a_i$'s --  that are
required to make the above map injective.
In
\cite{balan_signal_2006}, the authors showed that in the real case ($x
\in \mathbb{R}^d$), at least $2d-1$ such measurements
 are necessary and generically sufficient to guarantee injectivity, while in the complex case a generic sample size of $m \geq 4d-2$ suffices. Here generic is to be understood in the sense that the sets of measurements of such size which do not allow for recovery form an algebraic variety in the space of all frames.
Also, the latter bound is close to optimal: as shown in \cite{heinosaari_quantum_2013}, it follows from the results derived in \cite{sanderson_immersions_1964} that a sample size of $m \geq \left( 4+ o(1) \right)d$ is necessary (cf.~\cite{MixonBlog}).
However, finding the precise bound is still an open problem.

Balan et al.~\cite{balan_painless_2009} consider the scenario of $\mathcal{O}(d^2)$ measurements, which form a complex projective $2$-design (cf.~Def.~\ref{def:design} below).  They derive an explicit reconstruction formula for this setup based on the following observation well known in conic programming.
Namely, the quadratic constraints on $x$ are linear in the
outer product $x x^*$: 
\begin{equation}\label{eqn:linearized}
	y_i = |\langle a_i, x \rangle|^2 = \tr \left( (a_i a_i^*) (x x^*) \right).
\end{equation}
This ``lifts'' the problem to matrix space of dimension $d^2$,
where it becomes linear and can be explicitly solved to find the unique solution.

As we will show in Theorem \ref{thm:converse_bound}, it is, without making additional assumptions on the $2$-design, not possible to use as measurements a random subset of this $2$-design which is of size $o(d^2)$. In other words, for the measurement scenario described in \cite{balan_painless_2009}, the quadratic scaling in $d$ is basically unavoidable.

To contrast these two extreme approaches, ref.~\cite{balan_signal_2006} works with a number of measurements close to the absolute minimum, but there are no tractable reconstruction schemes provided, the question
of numerical stability is not considered, and it is unclear whether
non-generic measurements -- i.e., vectors with additional structural
properties -- can be employed. On the other hand, the number of measurements in~\cite{balan_painless_2009} is much larger, while the measurements are highly structured and there is an explicit reconstruction method.
A number of recent works including this paper aim to balance between these two approaches, working with a number of measurements only slightly larger while having at least some of the desired properties mentioned above.

Ref.~\cite{alexeev_phase_2012} introduces a reconstruction method 
called \emph{polarization}
that works for $\mathcal{O}(d \log d)$ measurements and can handle
structured measurement vectors, including the {\em masked illumination}
setup that appears in diffraction imaging
\cite{bandeira_phase_2013}, where the measurements are generated by the discrete Fourier
transform preceded by a random diagonal matrix.
For Gaussian measurements, the polarization approach has also shown to be stable with respect to measurement noise \cite{alexeev_phase_2012}.
While simulations seem to suggest stability also for the derandomized masked illumination setup, a proof of stability is -- to our knowledge --
not available yet. 

An alternative approach, which we will also follow in this paper, is the {\em PhaseLift} algorithm, which is
based on the lifted formulation \eqref{eqn:linearized}. The algorithm was introduced in \cite{candes_phase_2013} and
 reconstruction guarantees have been provided in
\cite{candes_phaselift_2012, candes_solving_2012}.
The central observation
is that the matrix $x x^*$, while unknown, is certainly of rank one. 
This connects the phase retrievel problem with the young but already
extensive field of \emph{low-rank matrix recovery} \cite{recht_guaranteed_2010,
candes_power_2010, gross_recovering_2011, liu_universal_2011}.
Over the past years, this research program has rigorously
identified many instances in which low-rank matrices can be
efficiently reconstructed from few linear measurements.
The existing results on low-rank matrix recovery were not directly applicable to 
phase retrieval, because the measurement matrices $a_i a_i^*$
failed to be sufficiently \emph{incoherent} in the sense of
\cite{candes_power_2010, gross_recovering_2011} (the incoherence parameter captures the
well-posedness of a low-rank recovery problem).
For the case of Gaussian measurement vectors $a_i$, 
Cand\`es, Strohmer, Voroninski and Li were able to circumvent this problem, providing problem-specific stable recovery guarantees
\cite{candes_phaselift_2012, candes_solving_2012} 
for a number of measurements of optimal order
$\mathcal{O}(d)$. For recovery, they use a convex relaxation of the rank minimization problem, which makes the reconstruction algorithm tractable.

It should be noted, however, that because of the significantly increased problem dimensions, PhaseLift is not as efficient as many phase retrieval algorithms developed over the last decades in the physics literature (such as \cite{F82}) and the optimization literature (for example \cite{BCR03}). 
Recently there have been attempts to provide recovery guarantees for alternating minimization algorithms \cite{NJS13}, which are somewhat closer to the algorithms used in practice, but this direction of research is only at its beginnings.

While the above mentioned recovery guarantees for PhaseLift address
the issues of tractable reconstruction and stability with respect to
noise, these results leave open the question of whether measurement
systems with additional structure and less randomness still allow for
guaranteed recovery. There are both practical and theoretical
motivations for pursuing such generalizations: A practitioner may be
constrained in the choice of measurements by the application at hand
or reduce the amount of randomness required for implementation
purposes. The most prominent example are again masked Fourier
measurements, which appear as a natural model in diffraction imaging,
but a lot of different scenarios imposing different structure are
conceivable.  From a theoretical point of view, the use of Gaussian
vectors obscures the specific properties that make phase retrieval
possible. As discussed in the following subsection, it is a common
thread in randomized signal processing that results are first
established for Gaussian measurements and later generalized to
structured ensembles.

A different direction of research, which will not be pursued in this
paper, is to ask how additional structural assumptions on the signal
to be recovered, such as sparsity, can be incorporated into the
theory. A general analysis based on the Gaussian width of how many
measurements are needed to allow for stable recovery of a signal known
to lie in a set $T\subset {\mathbb{R}}^d$  is provided in \cite{EM12}.
Notably the results allow for measurements with arbitrary subgaussian
rather than just Gaussian entries. Efficient algorithms for recovery,
however, are not provided.  For the case of $s$-sparse signals, also
tractable recovery algorithms are available: It has been shown in
\cite{li_sparse_2012} that PhaseLift can recover $x$ with high
probability from Gaussian measurements for a number of measurements
$m$ proportional to $s^2$ (up to logarithmic factors), which, for
small $s$, can be considerably less than the dimension. In
\cite{EFS13}, it is shown that only a number of subgaussian
measurements scaling linearly in the sparsity (up to logarithmic
factors) is needed if recovery proceeds using certain greedy
algorithms. 

\subsection{Designs as a general-purpose tool for de-randomization}

In this paper, we focus on the theoretical aspect: which 
properties of a measurements are sufficient for PhaseLift to succeed?
We prove recovery guarantees for ensembles of measurement vectors
drawn at random from a finite set whose first $2t$ moments agree with those of Haar-random vectors (or,
essentially, Gaussian vectors). A configuration of finite vectors
which gives rise to such an ensemble is known as a \emph{complex
projective $t$-design}\footnote{
	The definition of a $t$-design varies
	between authors. 
	In particular, what is called a $t$-design here (and in most of the physics
	literature), would sometimes be referred to as a $2t$ or
	even a $(2t+1)$-design.
	See Section~\ref{sec:complex
	projective designs}
	for our precise definition.
}.
Designs were introduced 
by Delsarte, Goethals and Seidel
in a seminal paper \cite{delsarte_spherical_1977}
and have been studied in
algebraic combinatorics
\cite{sidelnikov_spherical_1999}, coding theory \cite{delsarte_spherical_1977,nebe_invariants_2001}, and recently
in quantum information theory \cite{scott_tight_2006,
ambainis_quantum_2007,hayashi2005reexamination, gross_evenly_2007,
brandao_local_2012}.
Furthermore, complex projective $2$-designs were the key ingredient
for the reconstruction formula for phase retrieval proposed in \cite{balan_painless_2009}. 

One may see a more general philosophy behind this approach. In the
field of sparse and low-rank reconstruction, a number of recovery
results had first been established for Gaussian measurements. In
subsequent works, it has then been proven that measurements drawn at
random from certain fixed orthonormal bases are actually sufficient.
Examples include uniform recovery guarantees for compressed sensing
(\cite{CT05:Decoding, badadewa08} vs.~\cite{carota06-1, RV08:sparse})
and low-rank matrix recovery (\cite{recht_guaranteed_2010}
vs.~\cite{liu_universal_2011}), respectively. Typically, the
de-randomized proofs require  much higher technical efforts and
deliver slightly weaker results. For a recent survey on structured random measurements in signal processing see \cite{krahmer_structured_2014}.

As the number of measurements needed for phase retrieval is larger
than the signal space dimension, one cannot expect these results to
exactly carry over to the phase retrieval setting. Nevertheless, the
question remains whether there is a larger, but preferably not too
large, set such that measurements drawn from it uniformly at random allow
for phase retrieval reconstruction guarantees. In some sense, the
sampling scenario we seek can be interpreted as an
interpolation between the maximally random setup of Gaussian
measurement with an optimal order of measurements and the construction
in \cite{balan_painless_2009}, which is completely deterministic, but
suboptimal in terms of the embedding dimension. While in this paper,
we will focus on the phase retrieval problem, we remark that such an
interpolating approach between measurements drawn from a basis and
maximally random measurements may also be of interest in other
situations where constructions from bases are known, but lead to
somewhat suboptimal embedding dimensions.

 The concept of $t$-designs, as defined in Section~\ref{sub:complex_designs}, provides such an interpolation. The intuition behind that definition is that with growing $t$, more and more moments of the random vector corresponding to a random selection from the $t$-design agree with the Haar measure on the unit sphere.
In that sense, as $t$ scales up further, $t$-designs give better and better
approximations to Haar-random vectors.

The utility of this concept as a general-purpose de-randomization tool
for Hilbert-space valued random construtions has been appreciated for
example in quantum information theory \cite{ambainis_quantum_2007,
low_large_2009}. It has been compared 
\cite{ambainis_quantum_2007}
to the notion of \emph{$t$-wise
independence}, which plays a role for example in the analysis of
discrete randomized algorithms \cite{luby_pairwise_2006},
seems to have been long appreciated in coding theory.
The smallest $t$-design in $\CC^d$
consists of $\mathcal{O}(d^{2t})$ elements. Thus, whenever that lower
bound is met, drawing a single element from a design requires $2t\log
d$ bits, as opposed to $2d$ bits for a complex Bernoulli vector -- an
exponential gap.

From a practical point of view, the usefulness of these concepts
hinges on the availability of constructions for designs. 
Explicit constructions
for any order $t$ and any dimension $d$ are known
\cite{hayashi2005reexamination, bajnok1992construction,
korevaar_chebyshev_1994,seymour_averaging_1984} -- however, they 
are  typically ``inefficient''
in the sense that they require a vector set of exponential size.
For example, the construction in \cite{hayashi2005reexamination} uses
$\mathcal{O}(t)^d$ vectors which is exponential in the dimension $d$. 

Tighter analytic
expressions for \emph{exact}
designs are notoriously difficult to find. Designs of degree 2 
are widely known
\cite{schwinger_unitary_1960,zauner_quantendesigns_1999, konig_cubature_1999,
klappenecker_mutually_2005}.
A concrete example is used for the converse bound
in Section~\ref{sec:converse_bounds} (as well as for the converse
bounds for low-rank matrix recovery from Fourier-type bases in
\cite{gross_recovering_2011}).
For degree 3, both real\footnote{
	While stated only for dimensions that are a power of $2$, the
	results can be used for construtions in arbitrary dimensions
	\cite{kueng_stabilizer_2013}.
}
\cite{sidelnikov_spherical_1999} and complex 
\cite{kueng_stabilizer_2013} designs are known. 
For higher $t$, there are numerical methods based on the notion of the
\emph{frame potential} \cite{renes_symmetric_2004,
klappenecker_mutually_2005,  kueng_stabilizer_2013}
, non-constructive
existence proofs \cite{seymour_averaging_1984}, and constructions in
sporadic dimensions (c.f.\ \cite{bachoc_modular_2001} and references
thererin).

Importantly, almost-tight randomized constructions for
\emph{approximate designs} for arbitrary degrees and dimensions are
known \cite{ambainis_quantum_2007,
hayashi2005reexamination,
brandao_local_2012}.
The simplest results \cite{hayashi2005reexamination} show that collections of Haar-random
vectors form approximate
$t$-designs. This indeed can reduce randomness: One only needs to
expend a considerable amount of randomness \emph{once} to generated a
design -- for subsequent applications it is sufficient to sample small
subsets from it\footnote{ The situation is comparable to the use of
random graphs as randomness expanders \cite{hoory_expander_2006}.}.
Going further, there have been recent deep results on designs obtained
from certain structured ensembles \cite{brandao_local_2012}. We do not
describe the details here, as they are geared toward quantum problems
and may have to be substantially modified to be applicable to the
phase retrivial. The only connection to phase retrieval to date is the
estimation of pure quantum states \cite{heinosaari_quantum_2013,
mondragon_determination_2013}.

Finally we point out that the notion of the \emph{frame potential} above
is no coincidence. In \cite{bachoc_tight_2012} a frame-theoretic 
approach to designs is provided, underlining their close connection.

\subsection{Main results}

In this paper, we show that spherical designs can indeed be used to
partially derandomize recovery guarantees for underdetermined
estimation problems; we generalize the recovery guarantee in \cite{candes_phaselift_2012} to measurements drawn uniformly at random from complex projective designs, at the cost of a slightly higher number of measurements.

\begin{theorem}[Main Theorem]		\label{thm:main_theorem}
	Let $x \in \CC^d$ be the unknown signal. 
	Suppose that $\|
	x\|_{\ell_2}^2$ is known and that $m$ measurement vectors $a_1,
	\dots, a_m$ have been sampled
	independently and
	uniformly at random
	from a $t$-design $D_t \subset \CC^d$  ($t \geq 3$). Then,
	with
	probability at least
	$1-\mathrm{e}^{-\omega}$, PhaseLift
	(the convex optimization problem (\ref{eq:convex_program}) below) recovers
	$x$ up to a global phase, provided that the
	sampling rate exceeds
	\begin{equation}
		m \geq \omega \,Ct\, d^{1+2/t} \log ^2 d. 
	\end{equation}
	Here $\omega \geq 1$ is an arbitrary parameter and $C$ is a
	universal constant.
\end{theorem}

As the discussion of the previous subsection suggests, the bounds on the
sampling rate decrease as the order of the design increases. For fixed
$t$, and up to poly-log factors, it is proportional to
$\mathcal{O}(d^{1+2/t})$. 
This is sub-quadratic for the regime $t\geq 3$ where our arguments
apply.
If the degree  is allowed to grow
logarithmialy with the dimension (as $t=2\log d$), we recover an
optimal, linear scaling up to a polylog
overhead, $m=\mathcal{O}(d\,\log^3 d)$.

In light of the highly structured, analytical and exact designs known
for degree 2 and 3, it is of great interest to ask whether a linear scaling can already be achieved for some small, fixed $t$. As shown by the following theorem, however, for $t=2$ not even a subquadratic scaling is possible if no additional assumptions are made, irrespective of the reconstruction algorithm used. 

\begin{theorem}[Converse bound] \label{thm:converse_bound}
	Let $d$ be a prime power. Then there exists a $2$-design $D_2\subset
	\CC^d$ and orthogonal, normalized vectors $x, z\in\CC^d$ 
	which have the following property.
	
	Suppose that 
	$m$ measurement vectors $y_1,
	\dots, y_m$ are sampled
	independently and
	uniformly at random from $D_2$.
	Then, for any $\omega \geq 0$, the number of measurements must obey
	\begin{equation*}
		m \geq \frac{\omega}{4}d(d+1), 
	\end{equation*}
	or the event 
	\begin{equation*}
		|\langle a_i, x\rangle|^2 =
		|\langle a_i, z\rangle|^2
		\quad
		\forall\, i\in\{1, \dots, m\}
	\end{equation*}
	will occur
	with probability at
	least $\mathrm{e}^{-\omega}$. 
\end{theorem}

It is worthwhile to put this statement in perspective with other advances in the field. 
Throughout our work, we have only demanded that the set of all possible measurement vectors forms a $t$-design
and have not made any further  assumptions. 
Theorem \ref{thm:converse_bound} has to be interpreted in this regard: 
The 2-design property \emph{alone} does not allow for a sub-quadratic scaling when a ``reasonably small'' probability of failure is required in the recovery process.

Note that this does not exclude the possibility that certain realizations of 2-designs can perform better, if additional structural properties can be exploited. 
A good example for such a measurement process is the multi-illumination setup provided in \cite{candes_masked_2013}.
In \cite{gross_improved_2014} the authors of this paper verified that the set of all measurement vectors used in the framework of \cite{candes_masked_2013} does constitute a 2-design (Lemma 6). 
Additional structural properties
 -- most notably a certain correlated Fourier basis structure in the individual measurements -- 
allowed for establishing recovery guarantees already for $m = \mathcal{O} \left( d \log^4 d \right)$ measurements \cite{candes_masked_2013} and $m = \mathcal{O}(d \log^2 d)$ \cite{gross_improved_2014}, respectively 
-- which both clearly are sub-quadratic sampling rates.

\subsection{Outlook}

There are a number of  problems left open by our analysis. First,
recall that our results achieve linear scaling up to logarithmic
factors only when samples are drawn from a set of superpolynomial
size. Thus it would be very interesting to find out whether there are
polynomial size sets such that sampling from them achieves such a
scaling, in particular, if $t$-designs for some fixed $t$ can be used.
The case of $t=3$ seems particularly important in that regard, 
since the converse bound (Theorem \ref{thm:converse_bound})
shows that a design order of at least 3 is necessary.
Also,
highly structued 3-designs are known to exist (see above).

Another important follow-up problem concerns approximate $t$-designs.
While our main result is phrased for exact $t$-designs, certain
scenarios will only exhibit approximate design properties. We expect
that our proofs can be generalized to such a setup, but also leave
this problem for future work. Lastly, the reconstruction quality for
noisy measurements is also an important issue yet to be investigated.

\section{Numerical Experiments}

In this section we complement our theoretical results with numerical experiments, which we have implemented in Matlab using
CVX \cite{cvx1, cvx2}. As may have been expected, these experiments suggest that
PhaseLift from designs
actually works much better than our main theorem suggests. 
To be
concrete, we use \emph{stabilizer states} -- a highly structured
vector set which is very prominent in quantum information theory
\cite{gottesman_stabilizer_1997, gottesman_heisenberg_1999}.
 Stabilizer states exist in any dimension, though their
properties are somewhat better-behaved in prime power dimensions.
In this case, there exists
$\mathcal{O}(d^{\log d})$ stabilizer state vectors.  
Due to their rich
combinatorial structure, these vectors can be constructed efficiently.
For dimensions $d=2^n$ that are a power of two, it is known 
\cite{kueng_stabilizer_2013} 
that the
set of stabilizer states forms a
3-design.
This statement is false for other prime power
dimensions ($d \neq 2^n$ for some $n$), where they only form an exact
2-design.  However, weighted 3-designs can be constructed for
arbitrary dimensions $d$ by projecting down stabilizer states from the
next largest power-of-2-dimension $2^n$ obeying $2^{n-1}~<d<2^n$
\cite{kueng_stabilizer_2013}. For further clarification of the concept
of exact and weighted $t$-designs we defer the reader to \cite{scott_tight_2006}
and references therin.

We have used these vectors in our numerical simulations, the results of
which are depicted in Figure~\ref{fig:numerics}. For each dimension $d$
between 1 and 32 ($x$-axis) and for each number of measurements $m$ ranging
from 1 to 160 ($y$-axis), we ran a total of 30 independent experiments.
Each such experiment consisted in choosing a Haar-random (normalized
Gaussian) complex vector $x$ as test signal. Then, we drew $m$
projected stabilizer states uniformly at random and calculated their
squared overlap with the test signal. We then ran PhaseLift on this
data and declared the recovery a ``success'' if the 
Forbenius distance between the 
reconstructed
matrix $\tilde X$ and the true projection $X = x x^*$ was smaller than
$10^{-3}$.
Figure~\ref{fig:numerics} depics the empirical success probability:
Black corresponds to only failures, white to only successes.

We obtain the picture of a relatively sharp  phase transition along  a
line that scales linearly in the problem dimension.
In fact,
the transtion seems to occur in the vicinity of the line $m=4d-4$ --
drawn in red in Figure~\ref{fig:numerics}.  
This seems to agree with the conjecture that $4d-4$ 
measurements are required for injectivity  (see
e.g.~\cite{heinosaari_quantum_2013}).
However, there are a few differences in the problem setup: Firstly, the conjecture only asks whether there is a unique solution, while the numerical simulations study whether the PhaseLift algorithm can find it. Secondly, the conjecture concerns unique solutions for all possible inputs, while numerically, we estimate the success probability. And thirdly, the conjecture states that generic measurements work, while our simulations use a specific random procedure (drawn uniformly from a 3-design) to generate them.

%

\begin{figure}[h] 
\centering
\includegraphics[width=0.9\textwidth]{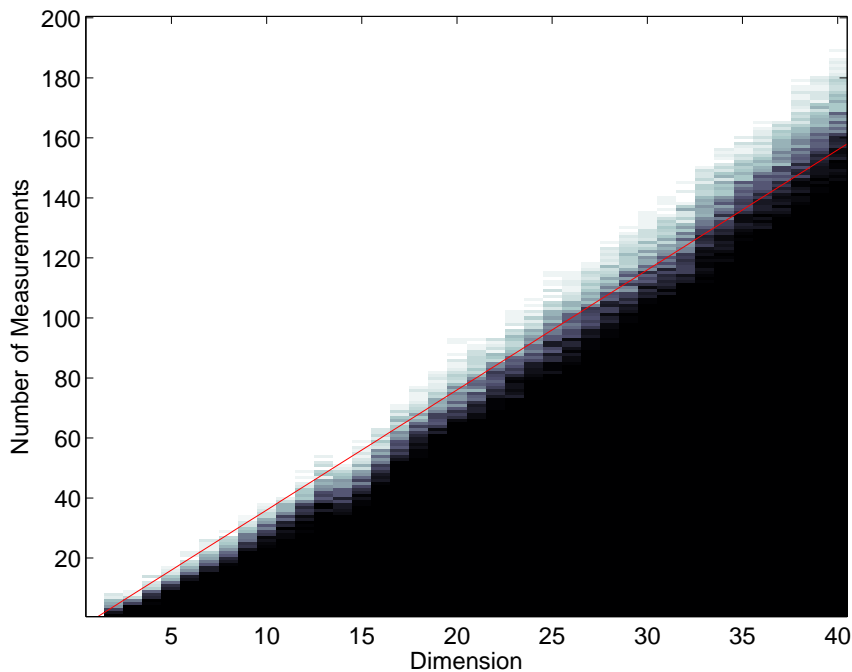}
\caption{Phase Diagram for PhaseLift from (projected) stabilizer
states, which form an exact 3-design in dimensions $2^n$ and a
weighted one else.
The $x$-axis indicates the problem's dimension, while the $y$-axis denotes the number of independent design measurements performed. 
The frequency of a successful recovery over 30 independent runs of the experiment appears color-coded from black (zero) to white (one).
To guide the eye, we have furthermore included a red line indicating $m = 4d - 4$. 
}\label{fig:numerics}
\end{figure}

\section{Technical Background and Notation}

\subsection{Vectors, Matrices and matrix valued Operators}

In this work we require three different objects of linear algebra:
vectors,  matrices and operators acting on matrices.

We will work with vectors in a $d$-dimensional complex Hilbert space $V^d$ equipped with an inner product $\langle \cdot, \cdot \rangle$. We refer to the associated induced norm by
\begin{equation*}
\| z \|_{\ell_2} = \sqrt{\langle z, z \rangle } \quad \forall z \in V^d.
\end{equation*}
We will denote such vectors by latin characters.
For $z \in V^d$, we define the dual vector $z^* \in (V^d)^*$ via
\begin{equation*}
z^* y = \langle z, y \rangle \quad \forall y \in V^d.
\end{equation*}

On the level of matrices we will exclusively consider $d \times d$ dimensional hermitian matrices, which we denote by capital latin characters.
Endowed with the Hilbert-Schmitt (or Frobenius) scalar product
\begin{equation}
(Z,Y) = \tr (Z Y) ,	\label{eq:frobenius}
\end{equation}
the space $H^d$ becomes a Hilbert space. 
In addition to that, we will require the 3 different Schatten-norms
\begin{eqnarray*}
\| Z \|_1 &=& \tr (|Z|)\quad \textrm{(trace norm)}, 	\\
\|Z \|_2 &=& \sqrt{\tr (Z^2)} \quad \textrm{(Frobenius norm)},	\\
\| Z \|_\infty &=& \sup_{y \in V^d} \frac{ \|Zy\|_{\ell_2}}{\|y\|_{\ell_2}}
\quad \textrm{(operator norm)},
\end{eqnarray*}
where the second one is induced by the scalar product (\ref{eq:frobenius}).
These three norms are related via the inequalities
\begin{equation*}
\| Z \|_2 \leq \| Z \|_1 \leq \sqrt{d} \| Z \|_2
\quad \textrm{and} \quad
\| Z \|_\infty \leq \| Z \|_2 \leq \sqrt{d} \| Z \|_\infty \quad \forall Z \in H^d.
\end{equation*}

We call a hermitian matrix $Z$ positive-semidefinite ($Z \geq 0$), 
if $ \langle y, Zy \rangle \geq 0$ for all $ y \in V^d$. 
Positive semidefinite matrices form a cone \cite{barvinok_course_2002} (Chapter II,12), which induces a partial ordering of matrices. Concretely, for $Z,Y \in H^d$ we write
$ Y \geq Z$ if $Y-Z$ is positive-semidefinite ($Y-Z \geq 0$). 

In this work, the identity matrix $\Id$ and rank-1 projectors are of particular importance.  They are
positive semidefinite and any matrix of the latter kind can be decomposed as $Z = z
z^*$ for some $z \in V^d$. 
Up to a global phase, they correspond
to vectors $z \in V^d$. 
The most important cases
are the projection onto the unknown signal $x$ and onto the $i$th
measurement vector $a_i$ respectively. They will be denoted by
\begin{equation*}
	X = x x^* 
	\quad \textrm{and} \quad
	A_i = a_i a_i^*.
\end{equation*}

Finally, we will frequently encouter \emph{matrix-valued
operators} acting on the space $H^d$.  We label such objects with
capital caligraphic letters and introduce the operator norm
\begin{equation*}
\| \MM \|_\op = \sup_{Z \in H^d} \frac{ \| \MM Z \|_2}{\| Z \|_2}
\end{equation*}
induced by the Frobenius norm on $H^d$.
It turns out that only very few matrix-valued operators will appear
below. These are:  the identity map
\begin{eqnarray*}
\II: H^d & \to & H^d 	\\
Z & \mapsto & Z \quad \forall Z \in H^d
\end{eqnarray*}
and (scalar multiples of) projectors onto some matrix $Y \in H^d$. The latter corresponds to
\begin{eqnarray*}
\Pi_Y: H^d & \to& H^d 		\\
Z & \mapsto &  Y (Y,Z) = Y \tr (YZ)   \quad \forall Z \in H^d.	
\end{eqnarray*}
The operator
\begin{equation*}
\Pi_\Id: Z \mapsto \Id \tr (\Id Z) = \Id \tr (Z) \quad \forall Z \in H^d,
\end{equation*}
is a very important example for this subclass of operators.
Note that it is not a normalized projection, but $\frac{1}{d} \Pi_\Id$ is. 
  Indeed, for $Z \in H^d$ arbitrary
\begin{equation}
\left( d^{-1} \Pi_\Id \right)^2 Z = d^{-2} \Id \tr (\Id \Pi_\Id Z ) = d^{-2} \tr (\Id) \Id \tr (Z)  = d^{-1}  \Pi_\Id Z .	\label{eq:PrId1}
\end{equation}

The notion of positive-semidefiniteness directly translates to matrix valued operators.
Concretely, we call 
$\MM$  positive-semidefinite ($\MM \geq 0$)
if $ (Z, \MM Z ) \geq 0$ for all $Z \in H^d$. 
Again, this induces a partial ordering.
Like in the matrix case, we write  $\mathcal{N} \geq \MM$,  if  $\mathcal{N} - \MM \geq 0$.
It is easy to check that all the operators introduced so far are positive semidefinite
and in particular we obtain the ordering
\begin{equation}
0 \leq \Pi_\Id \leq d \II,	\label{eq:PrId}
\end{equation}
by using (\ref{eq:PrId1}). 

\subsection{Multilinear Algebra}

The properties of $t$-designs are most naturally stated in the
framework of ($t$-fold) tensor product spaces. This motivates
recapitulating some basic concepts of multilinear algebra
that are going to greatly simplify our
analysis later on.
The concepts presented here are standard and can be found in any
textbook on multilinear algebra.  Our presentation has been influenced
in particular by \cite{landsberg_tensors_2012, watrous_lecture_2011}.

Let $V_1, \ldots, V_k$ be (finite dimensional, complex) vector spaces, and let $V_1^*, \ldots, V_k^*$ be their dual spaces. 
A function
\begin{equation*}
f: V_1 \times \cdots \times V_k \to \CC
\end{equation*}
is \emph{multilinear}, if it is linear in each $V_i$, $i=1,\ldots, k$. 
We denote the space of such functions by $V_1^* \otimes \cdots \otimes V_k^*$ and call it the \emph{tensor product} of $V_1^*, \ldots, V_k^*$. 
Consequently, the tensor product $\left(V^d \right)^{\otimes k} = \bigotimes_{i=1}^k V^d $ is the space of all multilinear functions
\begin{equation}
f: \underset{k\textrm{ times }}{\underbrace{\left( V^{d}\right)^*\times\cdots\times \left(V^{d}\right)^*}}\mapsto \CC ,		\label{eq:tensor_vector}
\end{equation}
and we call the elementary elements $z_1 \otimes \cdots \otimes z_k$ the \emph{tensor product} of the vectors $z_1, \ldots, z_k \in V^d$. 
Such an element can alternatively be defined more concretely via the \emph{Kronecker product} 
of the individual vectors. 
However, such a construction requires an explicit choice of basis in $V^d$ which is not the case in (\ref{eq:tensor_vector}).

With this notation, the space of linear maps $V^d \to V^d$ ($d \times d$-matrices) corresponds to the tensor product $M^d:=V^d \otimes \left(V^d \right)^* $ which is spanned by 
$\left\{y \otimes z^*:\;y,z \in V^d \right\}$ -- the set of all rank-1 matrices. 
For this generating set of $ M^d$, we define the \emph{trace} to be the natural bilinear map
\begin{eqnarray*}
\tr: V^d \otimes \left(V^d \right)^* & \to & \CC 	\\
\left( y \otimes z^* \right) & \mapsto & z^* y = \langle z, y \rangle 
\end{eqnarray*}
for all $y,z \in V^d$. The familiar notion of trace is obtained by extending this definition linearly to $M^d$.

Using $M^d = V^d \otimes \left(V^d\right)^*$ allows us to define the (matrix) tensor product $\left(M^d \right)^{\otimes k}$ to be the space of all multilinear functions
\begin{equation*}
f: \underset{k\textrm{ times }}{\underbrace{\left( \left(V^d \right)^* \times V^d \right) \times \cdots \times \left( \left(V^d \right)^* \times V^d \right)}}\to \CC
\end{equation*}
in complete analogy to the above. We call the elements $Z_1 \otimes \cdots \otimes Z_k$ the tensor product of the matrices $Z_1, \cdots, Z_k \in M^d$. 

On this tensor space, we define the \emph{partial trace} (over the $i$-th system) to be
\begin{eqnarray*}
\tr_i: \left(M^d \right)^{\otimes k} & \to & \left(M^d \right)^{\otimes (k-1)}	\\
Z_1 \otimes \cdots \otimes Z_k & \mapsto & \tr(Z_i) \left( Z_1 \otimes \cdots \otimes Z_{i-1} \otimes Z_{i+1} \otimes \cdots \otimes Z_k \right).
\end{eqnarray*}
Note that with the identification $M^d = V^d \otimes \left(V^d \right)^*$, $\tr_i$ corresponds to the natural contraction at position $i$. 
The partial trace over more than one system can be obtained by concatenating individual traces of this form, e.g. for $1\leq i < j \leq k$
\begin{equation*}
\tr_{i,j}:= \tr_i \circ \tr_j: \left( M^d \right)^{\otimes k}  \to \left( M^d \right)^{\otimes (k-2)}.	
\end{equation*}
In particular, the \emph{full trace} then corresponds to
\begin{eqnarray*}
\tr:= \tr_{1,\ldots,k}: \left(M^d \right)^{\otimes k} & \to & \CC	\\
\left( Z_1 \otimes \cdots \otimes Z_k \right) &\mapsto& \tr (Z_1) \ldots \tr(Z_k).
\end{eqnarray*}

Let us now return to the tensor space $\left(V^d \right)^{\otimes k}$ of vectors. 
We define the (symmetrizer) map 
$P_{\Sym^k}:  \left(V^d \right)^{\otimes k}  \to  \left( V^d \right)^{\otimes k} $
via their action on elementary elements:
\begin{equation}
P_{\Sym^k} \left( z_1 \otimes \cdots \otimes z_k \right) := \frac{1}{k!} \sum_{\pi \in S_k} z_{\pi (1)} \otimes \cdots \otimes z_{\pi(k)}, \label{eq:symmetrizer}
\end{equation}
where $S_k$ denotes the group of permutations of $k$ elements.
This map projects $ \left(V^d \right)^{\otimes k}$ onto the totally symmetric subspace $\Sym^k$ of $\left(V^d \right)^{\otimes k}$ whose dimension \cite{landsberg_tensors_2012} is
\begin{equation}
\dim \Sym^k = \binom{d+k-1}{k}		\label{eq:dimsym}.
\end{equation}

\subsection{Complex projective designs} \label{sub:complex_designs}
\label{sec:complex projective designs}

The idea of (real) spherical designs originates in coding theory \cite{delsarte_spherical_1977} and has been extended to more general spaces in
 \cite{neumaier_combinatorial_1981,hoggar_t_1982, levenshtein1998universal}. 
We refer the interested reader to Levenshtein \cite{levenshtein1998universal} for a unified treatment of designs in general metric spaces 
and from now on focus on designs in the complex vector space $ V^d$. 

Roughly speaking, a complex projective $t$-design is a finite subset of the complex unit sphere in $V^d$ with the property that the discrete average of any polynomial of degree $t$ or less equals its uniform average.
Many equivalent definitions -- see e.g. \cite{neumaier_combinatorial_1981, hoggar_t_1982, konig_cubature_1999} -- capture this essence. 
However, there is a more explicit definition of a $t$-design that is much more suitable for our purpose:

\begin{definition}[Definition 2 in \cite{scott_tight_2006}]	\label{def:design}
A finite set $\{w_1,\ldots, w_N\}\subset V^d$ of normalized vectors is called a \emph{$t$-design} of dimension $d$ if and only if 
\begin{equation}
\frac{1}{N}\sum_{i=1}^N (w_i w_i ^*)^{\otimes t}  =  \dim (\Sym^t )^{-1}  P_{\Sym^t},	\label{eq:design_def2}
\end{equation}
where $P_{\Sym^t}$ denotes the projector onto the totally symmetric subspace (\ref{eq:symmetrizer}) of $(V^d)^{\otimes t}$ and
consequently
$\dim (\Sym^t ) = \binom{d+t-1}{t}$.
\end{definition}

Note that the defining property (\ref{eq:design_def2}) is invariant
under global phase changes $w_i \mapsto \mathrm{e}^{i \phi} w_i$, thus
it matches the symmetry of the phase retrieval problem.
The definition above is equivalent to demanding
\begin{equation*}
	\frac{1}{N} \sum_{i=1}^N (w_i w_i^*)^{\otimes t} = 
	\int_{w} \mathrm{d}w\,(w w^*)^{\otimes t},	
\end{equation*}
where the right hand side is integrated with respect to the Haar
measure. This form makes the statement that $t$-designs mimic the first
$2t$ moments of Haar measure more explicit.

P. Seymor and T. Zaslavsky proved in \cite{seymour_averaging_1984} that $t$-designs on $V^d$ exist for every $t,d \geq 1$, provided that $N$ is large enough ($N \geq N(d,t)$),
but  they do not give an explicit construction. 
A necessary criterion -- cf. \cite{hoggar_t_1982, konig_cubature_1999} -- for the $t$-design property is that the number of vectors $N$ obeys 
\begin{equation}
N \geq \binom{d + \lceil t/2 \rceil -1}{\lceil t/2 \rceil} \binom{d+ \lfloor t/2 \lfloor - 1}{\lfloor t/2 \rfloor} = \mathcal{O}(d^{2t}).	\label{eq:design_bound}
\end{equation}

However, the proof in \cite{seymour_averaging_1984} is non-constructive and known constructions are ``inneficient'' in the sense that the number of vectors required greatly exceeds (\ref{eq:design_bound}). 
Hayashi et al. \cite{hayashi2005reexamination} proposed a construction requiring $ \mathcal{O}(t)^d$ vectors. 
For real spherical designs other ``inefficient'' constructions have been proposed \cite{bajnok1992construction, korevaar_chebyshev_1994} ($N = t^{\mathcal{O}(d^2)}$) which can be used to obtain complex projective designs.

Adressing this apparant lack of efficient constructions, Ambainis and Emerson \cite{ambainis_quantum_2007} proposed the notion of \emph{approximate desings}.
These vector sets only fulfill property (\ref{eq:design_def2}) only up to an $\epsilon$-precision, 
but their great advantage is that they can be constructed efficiently. Concretely, they show that for every $d \geq 2t$, there exists an $\epsilon = \mathcal{O}(d^{-1/3})$
approximate $t$-design consisting of $\mathcal{O}(d^{3t})$ vectors only.

The great value of $t$-designs is due to the following fact:
If we sample $m$ vectors $a_i, \ldots, a_m$ iid from a $t$-design $D_t = \left\{ w_1, \ldots, w_N \right\}$, the design property guarantees (with $A_i = a_i a_i^*$ and $W_i = w_i w_i^*$)
\begin{equation*}
\EE \left[ \frac{1}{m} \sum_{i=1}^m A_i^{\otimes k} \right] 
= \EE \left[ A_1^{\otimes k} \right] 
= \frac{1}{N} \sum_{i=1}^N W_i^{\otimes k} 
= \binom{d+k-1}{k}^{-1} P_{\Sym^k} 
\end{equation*}
for all $1 \leq k \leq t$. This knowledge about the first $t$ moments of the sampling procedure is the key ingredient for our partial derandomization of Gaussian PhaseLift \cite{candes_phaselift_2012}.

\subsection{Large Deviation Bounds}

This approach makes heavy use of operator-valued large deviation
bounds. They have been established first in the field of quantum
information by Ahlswede and Winter \cite{ahlswede_strong_2002}. Later
the first author of this paper and his coworkers successfully applied
these methods to the problem of low rank matrix recovery
\cite{gross_recovering_2011, gross_quantum_2010}.  By now these methods are widely used
and we borrow them in their most recent (and convenient) form from Tropp
\cite{tropp_user_2012,tropp_introduction_2012}.

\begin{theorem}[Uniform Operator Bernstein inequality, \cite{tropp_user_2012,gross_recovering_2011}] \label{thm:bernstein}
Consider a finite sequence $\left\{M_k \right\}$ of independent, random self-adjoint  operators.
Assume that each random variable satisfies $\EE \left[ M_k \right] = 0$ and $\|M_k \|_\infty \leq \overline{R}$ (for some finite constant $\overline{R}$) almost surely
and define the norm of the total variance $ \sigma^2 := \| \sum_k \EE \left[ M_k^2 \right] \|_\infty$. Then the following chain of inequalities holds for all $t\geq 0$.
\begin{equation*}
\Pr \left[ \| \sum_k M_k \|_\infty \geq t \right] \leq d \; \exp \left( - \frac{t^2/2}{\sigma^2 + \overline{R}t/3} \right) 
\leq 
\begin{cases}
 d \, \exp ( - 3t^2 / 8 \sigma^2 ) & t \leq \sigma^2/\overline{R}	\\ 
d \, \exp (-3t/8\overline{R}) & t \geq \sigma^2/\overline{R} .
\end{cases}
\end{equation*}
\end{theorem} 

\begin{theorem}[Smallest Eigenvalue Bernstein Inequality, \cite{tropp_introduction_2012}]	\label{thm:smallest_eigenvalue_bernstein}
Let $S = \sum_k M_k$ be a sum of iid random matrices $M_k$ which obey $\EE \left[ M_K \right] = 0$ and $\lambda_{\textrm{min}}(M_k) \geq - \underline{R}$ almost surely for some fixed $\underline{R}$.
With the variance parameter
$\sigma^2 (S) = \| \sum_k \EE \left[ M_k^2 \right] \|_\infty $
the following chain of inequalities holds for all $t \geq 0$.
\begin{equation*}
\Pr \left[ \lambda_{\min} (S) \leq - t \right] \leq d \exp \left( - \frac{t^2/2}{\sigma^2 + \underline{R}t/3} \right) 
\leq
\begin{cases}
 d \, \exp ( - 3t^2 / 8 \sigma^2 ) & t \leq \sigma^2/\underline{R}\\ 
d \, \exp (-3t/8 \underline{R}) & t \geq \sigma^2/\underline{R}. 
\end{cases}
\end{equation*}
\end{theorem}

\subsection{Wiring Diagrams}	\label{sub:wiring}

The defining property 
(\ref{eq:design_def2}) 
of $t$-designs is phrased in terms of tensor spaces. To work with
these notions practically, we need tools for efficiently computing
contractions between high-order tensors.
The concept of \emph{wiring diagrams} provides such a method -- see
\cite{landsberg_tensors_2012} for an introduction and also
\cite{turaev_quantum_1994,cvitanovic_group_1984} (however, they use a
slightly different notation). Here, we give a brief description that
should suffice for our calculations.

Roughly, the calculus of wiring diagrams associates with every tensor
a box, and with every index of that tensor a line emanating from the
box. Two connected lines represent contracted indices. (More
precisely, we place contravariant indices of a tensor on top of the
associated box and covariant ones at the bottom. However, one should
be able to digest our calculations without reference to this detail).
A matrix $A: V^d \to V^d$ can be seen as a two-indexed tensor
${A^i}_j$. It will thus be represented by a node
$
\tikz[heighttwo,xscale=.5,baseline]{
\coordinate(up)at(0,0.8);
\coordinate(mid)at(0,0.5);
\coordinate(down)at(0,0.2);
\draw(down)to(up);
\draw(mid)node[vector]{$A$};
}
$
with the upper line corresponding to the index $i$ and the lower one
to $j$. Two matrices $A, B$ are multiplied by contracting $B$'s
``contravariant'' index with $A$'s ``covariant'' one:
\begin{equation*}
	{(A B)^i}_j = \sum_k {A^i}_k {B^k}_j
\end{equation*}
Pictographically, we write
\begin{equation*}
	AB
	=
	\tikz[heighttwo,xscale=.5,baseline]{
	\coordinate(up)at(0,01);
	\coordinate(mid1)at(0,0.73);
	\coordinate(mid2)at(0,0.27);
	\coordinate(down)at(0,0);
	\draw(down)to(up);
	\draw(mid1)node[vector]{$A$};
	\draw(mid2)node[vector]{$B$};
	}
\end{equation*}
The trace operation
\begin{equation*}
	A \mapsto \tr A = \sum_k {A^k}_k
\end{equation*}
corresponds to a contraction of the two indices of a matrix:
\begin{equation*}
	\tr (A) 
	= 
	\tikz[heighttwo,xscale=.5,baseline]{
	\coordinate(up)at(0,0.8);
	\coordinate(trup)at(0.5,0.8);
	\coordinate(mid)at(0,0.5);
	\coordinate(down)at(0,0.2);
	\coordinate(trdown)at(0.5,0.2);
	\draw(down)to(up);
	\draw(trdown)to(trup);
	\draw(up)to[out=90,in=90](trup);
	\draw(down)to[out=-90,in=-90](trdown);
	\draw(mid)node[vector]{$A$};
	}.
\end{equation*}
Tensor products are arranged in parallel:
\begin{equation*}
	A \otimes B
	=
	\tikz[heighttwo,xscale=.5,baseline]{
	\coordinate(up)at(0,0.9);
	\coordinate(upr)at(1,0.9);
	\coordinate(mid)at(0,0.5);
	\coordinate(midr)at(1,0.5);
	\coordinate(down)at(0,0.1);
	\coordinate(downr)at(1,0.1);
	\draw(down)to(up);
	\draw(downr)to(upr);
	\draw(mid)node[vector]{$A$};
	\draw(midr)node[vector]{$B$};
	}.
\end{equation*}
Hence, a partial trace takes the following form:
\begin{equation*}
\tr_2 \left( A \otimes B \right)
=
\tikz[heighttwo,xscale=.5,baseline]{
\coordinate(up)at(0,0.9);
\coordinate(upr)at(1,0.8);
\coordinate(mid)at(0,0.5);
\coordinate(midr)at(1,0.5);
\coordinate(down)at(0,0.1);
\coordinate(downr)at(1,0.2);
\coordinate(trup)at(1.5,0.8);
\coordinate(trdown)at(1.5,0.2);
\draw(down)to(up);
\draw(downr)to(upr);
\draw(trdown)to(trup);
\draw(downr)to[out=-90,in=-90](trdown);
\draw(upr)to[out=90,in=90](trup);
\draw(mid)node[vector]{$A$};
\draw(midr)node[vector]{$B$};
}\;.
\end{equation*}

The last ingredient we need are the \emph{transpositions}
$\sigma_{(i,j)}$ on $(V^d)^{\otimes t}$ which act by interchanging the
$i$th and the $j$th tensor factor.
For example
\begin{equation*}
	\sigma_{(1,2)} \left( x \otimes y \otimes \cdots \right)
	= y \otimes x \otimes \cdots,
\end{equation*}
with $x,y \in V^d$ arbitrary. 
Transpositions suffice, because they generate the full group of
permutations.
For $ \left( V^d \right)^{\otimes 2} $ we only have
\begin{equation*}
\underline{1} = 
\tikz[heighttwo,xscale=.5,baseline]{
\coordinate(up)at(0,0.9);
\coordinate(upr)at(1,0.9);
\coordinate(mid)at(0,0.5);
\coordinate(midr)at(1,0.5);
\coordinate(down)at(0,0.1);
\coordinate(downr)at(1,0.1);
\draw(down)to(up);
\draw(downr)to(upr);
}
\;\textrm{(trivial permutation)}
\quad
\textrm{and}
\quad
\sigma_{(1,2)}
=
\tikz[heighttwo,xscale=.5,baseline]{
\coordinate(up)at(0,0.9);
\coordinate(upr)at(1,0.9);
\coordinate(mid)at(0,0.5);
\coordinate(midr)at(1,0.5);
\coordinate(down)at(0,0.1);
\coordinate(downr)at(1,0.1);
\draw(down)to[wavyup](upr);
\draw(downr)to[wavyup](up);
},
\end{equation*} 
but for higher tensor systems more permutations can occur. 
Consequently, permutations act by interchanging different input and output lines
and the wiring diagram representation allows one to  keep track of this pictorially. 
In fact, only the input and output position of a line matters.
We can use diagrams to simplify expressions by disentangling the corresponding lines.
Take $\sigma_{(1,2)}$ on $\left( V^d \right)^{\otimes 2}$ as an example. 
Using wiring diagrams we can derive the standard result
\begin{equation*}
\sigma_{(1,2)}^2 
=
\tikz[heighttwo,xscale=.5,baseline]{
\coordinate(up)at(0,0.9);	\coordinate(upr)at(1,0.9);
\coordinate(down)at(0,0.1); \coordinate(downr)at(1,0.1);
\coordinate(mid)at(0,0.5); \coordinate(midr)at(1,0.5);
\draw(down)to[wavyup](midr);
\draw(midr)to[wavyup](up);
\draw(downr)to[wavyup](mid);
\draw(mid)to[wavyup](upr);
}
=
\tikz[heighttwo,xscale=.5,baseline]{
\coordinate(up)at(0,0.9);	\coordinate(upr)at(1,0.9);
\coordinate(down)at(0,0.1); \coordinate(downr)at(1,0.1);
\coordinate(mid)at(0,0.5); \coordinate(midr)at(1,0.5);
\draw(down)to[wavyup](up);
\draw(downr)to[wavyup](upr);
}
= \underline{1}
\end{equation*}
pictorially. 
We are now ready to prove some important auxiliary results.

\begin{lemma}	\label{lem:tr1sym2}
Let $A,B \in H^d$ be arbitrary. Then it holds that
\begin{equation}
\tr_2 \left( P_{\Sym^2} A \otimes B \right) 
= \frac{1}{2} \left( \tr (B) A + BA \right).	\label{eq:tr1sym2}
\end{equation}
\end{lemma}

We remark that in general,
\begin{equation*}
	P_{\Sym^2} \left( X\otimes Y \right) \neq \frac{1}{2} \left( X \otimes Y + Y \otimes X \right),
\end{equation*}
which is, in our experience, a common misconception.

\begin{proof}[Proof of Lemma \ref{lem:tr1sym2}]
The basic formula (\ref{eq:symmetrizer}) for $P_{\Sym^2}$ is given by
\begin{equation*}
P_{\Sym^2} = \frac{1}{2} \sum_{\pi \in S_2} \sigma_{\pi(1),\pi(2)} = \frac{1}{2} \left( \underline{1} + \sigma_{(1,2)} \right),
\end{equation*}
and the concepts from above allow us to translate this into the following wiring diagram:

\begin{equation*}
\tikz[heighttwo,xscale=.5,baseline]{
\coordinate(up)at(0,0.9);
\coordinate(upr)at(1,0.9);
\coordinate(mid)at(0,0.5);
\coordinate(midr)at(1,0.5);
\coordinate(down)at(0,0.1);
\coordinate(downr)at(1,0.1);
\draw(down)to(up);
\draw(downr)to(upr);
\draw[draw=black,fill=gray!10](-0.3,.25)rectangle(1.3,.75);\node[basiclabel]at(0.51,.48){$P_{\Sym^2}$};
}
=
\frac{1}{2} \left(
\tikz[heighttwo,xscale=.5,baseline]{
\coordinate(up)at(0,0.9);
\coordinate(upr)at(1,0.9);
\coordinate(mid)at(0,0.5);
\coordinate(midr)at(1,0.5);
\coordinate(down)at(0,0.1);
\coordinate(downr)at(1,0.1);
\draw(down)to(up);
\draw(downr)to(upr);
}
+
\tikz[heighttwo,xscale=.5,baseline]{
\coordinate(up)at(0,0.9);
\coordinate(upr)at(1,0.9);
\coordinate(mid)at(0,0.5);
\coordinate(midr)at(1,0.5);
\coordinate(down)at(0,0.1);
\coordinate(downr)at(1,0.1);
\draw(down)to[wavyup](upr);
\draw(downr)to[wavyup](up);
}
\right).
\end{equation*}
(Note that this operator acts on the full tensor space $ \left( V^d \right)^{\otimes 2}$,
hence in the wiring diagram it is represented by a two-indexed box.) 
Applying the graphical calculus yields
\begin{eqnarray*}
\tr_{2} \left( P_{\Sym^2} A \otimes B \right)
&=& 
\tikz[heighttwo,xscale=.5,baseline]{
	\coordinate(top)at(0,1.3){};	\coordinate(topr)at(1,1.3){};
	\coordinate(mid1)at(0,1){};	\coordinate(mid1r)at(1,1){};
	\coordinate(mid2)at(0,0.5){};	\coordinate(mid2r)at(1,0.5){};
	\coordinate(bot)at(0,0){};		\coordinate(botr)at(1,0){};
	\coordinate(toptr)at(1.5,1.3){};	\coordinate(bottr)at(1.5,0){};
	\coordinate(in)at(0,1.4){};
	\coordinate(out)at(0,-0.1){};
	\draw(bot)to(top);
	\draw(botr)to(topr);
	\draw(mid1)node[vector]{$A$};
	\draw(mid1r)node[vector]{$B$};
	\draw[draw=black,fill=gray!10](-0.3,.25)rectangle(1.3,.75);\node[basiclabel]at(0.51,.48){$P_{\Sym^2}$};
	\draw(bottr)to(toptr);
	\draw(topr)to[out=90,in=90](toptr);
	\draw(botr)to[out=-90,in=-90](bottr);
	\draw(top)to(in);
	\draw(out)to(bot);
}
=
\frac{1}{2}
\left(
\tikz[heighttwo,xscale=.5,baseline]{
	\coordinate(top)at(0,1.3){};	\coordinate(topr)at(1,1.3){};
	\coordinate(mid1)at(0,1){};	\coordinate(mid1r)at(1,1){};
	\coordinate(mid2)at(0,0.5){};	\coordinate(mid2r)at(1,0.5){};
	\coordinate(bot)at(0,0){};		\coordinate(botr)at(1,0){};
	\coordinate(toptr)at(1.5,1.3){};	\coordinate(bottr)at(1.5,0){};
	\coordinate(in)at(0,1.4){};
	\coordinate(out)at(0,-0.1){};
	\draw(bot)to(top);
	\draw(botr)to(topr);
	\draw(mid1)node[vector]{$A$};
	\draw(mid1r)node[vector]{$B$};
	\draw(bottr)to(toptr);
	\draw(topr)to[out=90,in=90](toptr);
	\draw(botr)to[out=-90,in=-90](bottr);
	\draw(top)to(in);
	\draw(out)to(bot);
}
+
\tikz[heighttwo,xscale=.5,baseline]{
	\coordinate(top)at(0,1.3){};	\coordinate(topr)at(1,1.3){};
	\coordinate(mid1)at(0,1){};	\coordinate(mid1r)at(1,1){};
	\coordinate(mid2)at(0,0.5){};	\coordinate(mid2r)at(1,0.5){};
	\coordinate(bot)at(0,0){};		\coordinate(botr)at(1,0){};
	\coordinate(toptr)at(1.5,1.3){};	\coordinate(bottr)at(1.5,0){};
	\coordinate(in)at(0,1.4){};
	\coordinate(out)at(0,-0.1){};
	\draw(bot)to[wavyup](mid1r);
	\draw(botr)to[wavyup](mid1);
	\draw(bottr)to(toptr);
	\draw(mid1r)to(topr);
	\draw(topr)to[out=90,in=90](toptr);
	\draw(botr)to[out=-90,in=-90](bottr);
	\draw(mid1)to(in);
	\draw(out)to(bot);
	\draw(mid1)node[vector]{$A$};
	\draw(mid1r)node[vector]{$B$};
}
\right)	
=
\frac{1}{2}
\left(
\tikz[heighttwo,xscale=.5,baseline]{
	\coordinate(top)at(0,1.3){};	\coordinate(topr)at(1,1.3){};
	\coordinate(mid1)at(0,1){};	\coordinate(mid1r)at(1,1){};
	\coordinate(mid2)at(0,0.5){};	\coordinate(mid2r)at(1,0.5){};
	\coordinate(bot)at(0,0){};		\coordinate(botr)at(1,0){};
	\coordinate(toptr)at(1.5,1.3){};	\coordinate(bottr)at(1.5,0){};
	\coordinate(in)at(0,1.4){};
	\coordinate(out)at(0,-0.1){};
	\draw(bot)to(top);
	\draw(botr)to(topr);
	\draw(mid1)node[vector]{$A$};
	\draw(mid1r)node[vector]{$B$};
	\draw(bottr)to(toptr);
	\draw(topr)to[out=90,in=90](toptr);
	\draw(botr)to[out=-90,in=-90](bottr);
	\draw(top)to(in);
	\draw(out)to(bot);
}
+
\tikz[heighttwo,xscale=.5,baseline]{
\coordinate(up)at(0,1.4);
\coordinate(mid1)at(0,1);
\coordinate(mid2)at(0,0.3);
\coordinate(down)at(0,-0.1);
\draw(down)to(up);
\draw(mid1)node[vector]{$A$};
\draw(mid2)node[vector]{$B$};
}
\right)		\\
 &=& 
\frac{1}{2} \left(  \tr (B) A + BA \right)	,	
\end{eqnarray*}
which is the desired result.
\end{proof}

Obviously, it is also possible to obtain (\ref{eq:tr1sym2}) by direct calculation. 
We have included such a calculation in the appendix (Section
\ref{sub:alternative_proof}) to demonstrate the complexity of direct
calculations as compared to graphical ones.

We conclude this section with the following slightly more involved result.

\begin{lemma}	\label{lem:P3_calculation}
Let $A,B,C \in H^d$ be arbitrary. Then it holds that
\begin{eqnarray}
& & \tr_{2,3} \left( P_{\Sym^3} A \otimes B \otimes C \right)	\label{eq:tr23sym3} \\
 &=& 
\frac{1}{6} \left(A \; \mathrm{tr}(B) \mathrm{tr}(C) + BA \; \mathrm{tr}(C) + CA \mathrm{tr}(B)
	+ A \; \mathrm{tr}(BC) + CBA + BCA \right).	\nonumber
\end{eqnarray}
\end{lemma}

The proof can in principle be obtained by evaluating all permutations of 3 tensor systems algebraically and taking the partial trace afterwards.
However, a pictorial calculation using wiring diagrams is much faster and more elegant.

\begin{proof}
For permutations of three elements, formula (\ref{eq:symmetrizer}) implies
\begin{equation*}
P_{\Sym^3} = \frac{1}{6} \sum_{\pi \in S_3} \sigma_{\pi (1),\pi(2),\pi(3)}
= \frac{1}{6} \left( \sigma_{1,2,3} + \sigma_{2,1,3} + \sigma_{3,2,1}
+ \sigma_{1,3,2} + \sigma_{2,3,1} + \sigma_{3,1,2} \right),
\end{equation*}
where. $\sigma_{2,1,3}(u \otimes v \otimes w ) = (v \otimes u \otimes w)$, etc.
This in turn allows us to write
\begin{eqnarray*}
	&&
  \tikz[heighttwo,xscale=.5,baseline]{
	\coordinate(top)at(1,1.3){};\coordinate(topm)at(2,1.3){};\coordinate(topr)at(3,1.3){};
	\coordinate(mid)at(1,1){};\coordinate(midm)at(2,1){};\coordinate(midr)at(3,1){};
	\coordinate(bot)at(1,0){};\coordinate(botm)at(2,0){};\coordinate(botr)at(3,0){};
	\coordinate(tr3up)at(3.5,1.3){}; \coordinate(tr3bot)at(3.5,0){};
	\coordinate(tr2up)at(1.5,1.3){}; \coordinate(tr2bot)at(1.5,0){};
	\coordinate(in)at(1,1.4){};
	\coordinate(out)at(1,-0.1){};
	\draw(bot)to(top){};
	\draw(botm)to(topm){};
	\draw(botr)to(topr){};
           \draw[draw=black,fill=gray!10](0.8,.35)rectangle(3.2,.75);\node[basiclabel]at(2,.5){$P_{\Sym^3}$};
	\draw(mid)node[vector]{$A$};
	\draw(midm)node[vector]{$B$};
	\draw(midr)node[vector]{$C$};
	\draw(topr)to[out=90,in=90](tr3up);
	\draw(botr)to[out=-90,in=-90](tr3bot);
	\draw(tr3bot)to[wavyup](tr3up);
	\draw(topm)to[out=90,in=90](tr2up);
	\draw(botm)to[out=-90,in=-90](tr2bot);
	\draw(tr2bot)to[wavyup](tr2up);
	\draw(out)to(bot);
	\draw(top)to(in);
} \\
&=&
\frac{1}{6} \left(
\tikz[heighttwo,xscale=.5,baseline]{
	\coordinate(top)at(1,1.3){};\coordinate(topm)at(2,1.3){};\coordinate(topr)at(3,1.3){};
	\coordinate(mid)at(1,1){};\coordinate(midm)at(2,1){};\coordinate(midr)at(3,1){};
	\coordinate(bot)at(1,0){};\coordinate(botm)at(2,0){};\coordinate(botr)at(3,0){};
	\coordinate(tr3up)at(3.5,1.3){}; \coordinate(tr3bot)at(3.5,0){};
	\coordinate(tr2up)at(1.5,1.3){}; \coordinate(tr2bot)at(1.5,0){};
	\coordinate(in)at(1,1.4){};
	\coordinate(out)at(1,-0.1){};
	\draw(bot)to[wavyup](mid);
	\draw(botm)to[wavyup](midm);
	\draw(botr)to[wavyup](midr);
	\draw(mid)to[wavyup](top);
	\draw(midm)to[wavyup](topm);
	\draw(midr)to[wavyup](topr);
	\draw(mid)node[vector]{$A$};
	\draw(midm)node[vector]{$B$};
	\draw(midr)node[vector]{$C$};
	\draw(topr)to[out=90,in=90](tr3up);
	\draw(botr)to[out=-90,in=-90](tr3bot);
	\draw(tr3bot)to[wavyup](tr3up);
	\draw(topm)to[out=90,in=90](tr2up);
	\draw(botm)to[out=-90,in=-90](tr2bot);
	\draw(tr2bot)to[wavyup](tr2up);
	\draw(out)to(bot);
	\draw(top)to(in);
}
+
\tikz[heighttwo,xscale=.5,baseline]{
	\coordinate(top)at(1,1.3){};\coordinate(topm)at(2,1.3){};\coordinate(topr)at(3,1.3){};
	\coordinate(mid)at(1,1){};\coordinate(midm)at(2,1){};\coordinate(midr)at(3,1){};
	\coordinate(bot)at(1,0){};\coordinate(botm)at(2,0){};\coordinate(botr)at(3,0){};
	\coordinate(tr3up)at(3.5,1.3){}; \coordinate(tr3bot)at(3.5,0){};
	\coordinate(tr2up)at(2.5,1.3){}; \coordinate(tr2bot)at(2.5,0){};
	\coordinate(in)at(1,1.4){};
	\coordinate(out)at(1,-0.1){};
	\draw(bot)to[wavyup](midm);
	\draw(botm)to[wavyup](mid);
	\draw(botr)to[wavyup](midr);
	\draw(mid)to[wavyup](top);
	\draw(midm)to[wavyup](topm);
	\draw(midr)to[wavyup](topr);
	\draw(mid)node[vector]{$A$};
	\draw(midm)node[vector]{$B$};
	\draw(midr)node[vector]{$C$};
	\draw(topr)to[out=90,in=90](tr3up);
	\draw(botr)to[out=-90,in=-90](tr3bot);
	\draw(tr3bot)to[wavyup](tr3up);
	\draw(topm)to[out=90,in=90](tr2up);
	\draw(botm)to[out=-90,in=-90](tr2bot);
	\draw(tr2bot)to[wavyup](tr2up);
	\draw(out)to(bot);
	\draw(top)to(in);
}
+
\tikz[heighttwo,xscale=.5,baseline]{
	\coordinate(top)at(1,1.3){};\coordinate(topm)at(2,1.3){};\coordinate(topr)at(3,1.3){};
	\coordinate(mid)at(1,1){};\coordinate(midm)at(2,1){};\coordinate(midr)at(3,1){};
	\coordinate(bot)at(1,0){};\coordinate(botm)at(2,0){};\coordinate(botr)at(3,0){};
	\coordinate(tr3up)at(3.5,1.3){}; \coordinate(tr3bot)at(3.5,0){};
	\coordinate(tr2up)at(1.5,1.3){}; \coordinate(tr2bot)at(1.5,0){};
	\coordinate(in)at(1,1.4){};
	\coordinate(out)at(1,-0.1){};
	\draw(bot)to[wavyup](midr);
	\draw(botm)to[wavyup](midm);
	\draw(botr)to[wavyup](mid);
	\draw(mid)to[wavyup](top);
	\draw(midm)to[wavyup](topm);
	\draw(midr)to[wavyup](topr);
	\draw(mid)node[vector]{$A$};
	\draw(midm)node[vector]{$B$};
	\draw(midr)node[vector]{$C$};
	\draw(topr)to[out=90,in=90](tr3up);
	\draw(botr)to[out=-90,in=-90](tr3bot);
	\draw(tr3bot)to[wavyup](tr3up);
	\draw(topm)to[out=90,in=90](tr2up);
	\draw(botm)to[out=-90,in=-90](tr2bot);
	\draw(tr2bot)to[wavyup](tr2up);
	\draw(out)to(bot);
	\draw(top)to(in);
}
+
\tikz[heighttwo,xscale=.5,baseline]{
	\coordinate(top)at(1,1.3){};\coordinate(topm)at(2,1.3){};\coordinate(topr)at(3,1.3){};
	\coordinate(mid)at(1,1){};\coordinate(midm)at(2,1){};\coordinate(midr)at(3,1){};
	\coordinate(bot)at(1,0){};\coordinate(botm)at(2,0){};\coordinate(botr)at(3,0){};
	\coordinate(tr3up)at(3.5,1.3){}; \coordinate(tr3bot)at(3.5,0){};
	\coordinate(tr2up)at(1.5,1.3){}; \coordinate(tr2bot)at(1.5,0){};
	\coordinate(in)at(1,1.4){};
	\coordinate(out)at(1,-0.1){};
	\draw(bot)to[wavyup](mid);
	\draw(botm)to[wavyup](midr);
	\draw(botr)to[wavyup](midm);
	\draw(mid)to[wavyup](top);
	\draw(midm)to[wavyup](topm);
	\draw(midr)to[wavyup](topr);
	\draw(mid)node[vector]{$A$};
	\draw(midm)node[vector]{$B$};
	\draw(midr)node[vector]{$C$};
	\draw(topr)to[out=90,in=90](tr3up);
	\draw(botr)to[out=-90,in=-90](tr3bot);
	\draw(tr3bot)to[wavyup](tr3up);
	\draw(topm)to[out=90,in=90](tr2up);
	\draw(botm)to[out=-90,in=-90](tr2bot);
	\draw(tr2bot)to[wavyup](tr2up);
	\draw(out)to(bot);
	\draw(top)to(in);
}
+
\tikz[heighttwo,xscale=.5,baseline]{
	\coordinate(top)at(1,1.3){};\coordinate(topm)at(2,1.3){};\coordinate(topr)at(3,1.3){};
	\coordinate(mid)at(1,1){};\coordinate(midm)at(2,1){};\coordinate(midr)at(3,1){};
	\coordinate(bot)at(1,0){};\coordinate(botm)at(2,0){};\coordinate(botr)at(3,0){};
	\coordinate(tr3up)at(3.5,1.3){}; \coordinate(tr3bot)at(3.5,0){};
	\coordinate(tr2up)at(2.5,1.3){}; \coordinate(tr2bot)at(2.5,0){};
	\coordinate(in)at(1,1.4){};
	\coordinate(out)at(1,-0.1){};
	\draw(bot)to[wavyup](midr);
	\draw(botm)to[wavyup](mid);
	\draw(botr)to[wavyup](midm);
	\draw(mid)to[wavyup](top);
	\draw(midm)to[wavyup](topm);
	\draw(midr)to[wavyup](topr);
	\draw(mid)node[vector]{$A$};
	\draw(midm)node[vector]{$B$};
	\draw(midr)node[vector]{$C$};
	\draw(topr)to[out=90,in=90](tr3up);
	\draw(botr)to[out=-90,in=-90](tr3bot);
	\draw(tr3bot)to[wavyup](tr3up);
	\draw(topm)to[out=90,in=90](tr2up);
	\draw(botm)to[out=-90,in=-90](tr2bot);
	\draw(tr2bot)to[wavyup](tr2up);
	\draw(out)to(bot);
	\draw(top)to(in);
}
+
\tikz[heighttwo,xscale=.5,baseline]{
	\coordinate(top)at(1,1.3){};\coordinate(topm)at(2,1.3){};\coordinate(topr)at(3,1.3){};
	\coordinate(mid)at(1,1){};\coordinate(midm)at(2,1){};\coordinate(midr)at(3,1){};
	\coordinate(bot)at(1,0){};\coordinate(botm)at(2,0){};\coordinate(botr)at(3,0){};
	\coordinate(tr3up)at(3.5,1.3){}; \coordinate(tr3bot)at(3.5,0){};
	\coordinate(tr2up)at(1.5,1.3){}; \coordinate(tr2bot)at(1.5,0){};
	\coordinate(in)at(1,1.4){};
	\coordinate(out)at(1,-0.1){};
	\draw(bot)to[wavyup](midm);
	\draw(botm)to[wavyup](midr);
	\draw(botr)to[wavyup](mid);
	\draw(mid)to[wavyup](top);
	\draw(midm)to[wavyup](topm);
	\draw(midr)to[wavyup](topr);
	\draw(mid)node[vector]{$A$};
	\draw(midm)node[vector]{$B$};
	\draw(midr)node[vector]{$C$};
	\draw(topr)to[out=90,in=90](tr3up);
	\draw(botr)to[out=-90,in=-90](tr3bot);
	\draw(tr3bot)to[wavyup](tr3up);
	\draw(topm)to[out=90,in=90](tr2up);
	\draw(botm)to[out=-90,in=-90](tr2bot);
	\draw(tr2bot)to[wavyup](tr2up);
	\draw(out)to(bot);
	\draw(top)to(in);
}		
\right)\\
&=& 
\frac{1}{6} \left(
\tikz[heighttwo,xscale=.5,baseline]{
	\coordinate(top)at(1,1.3){};\coordinate(topm)at(2,1.3){};\coordinate(topr)at(3,1.3){};
	\coordinate(mid)at(1,1){};\coordinate(midm)at(2,1){};\coordinate(midr)at(3,1){};
	\coordinate(bot)at(1,0){};\coordinate(botm)at(2,0){};\coordinate(botr)at(3,0){};
	\coordinate(tr3up)at(3.5,1.3){}; \coordinate(tr3bot)at(3.5,0){};
	\coordinate(tr2up)at(1.5,1.3){}; \coordinate(tr2bot)at(1.5,0){};
	\coordinate(in)at(1,1.4){};
	\coordinate(out)at(1,-0.1){};
	\draw(bot)to[wavyup](mid);
	\draw(botm)to[wavyup](midm);
	\draw(botr)to[wavyup](midr);
	\draw(mid)to[wavyup](top);
	\draw(midm)to[wavyup](topm);
	\draw(midr)to[wavyup](topr);
	\draw(mid)node[vector]{$A$};
	\draw(midm)node[vector]{$B$};
	\draw(midr)node[vector]{$C$};
	\draw(topr)to[out=90,in=90](tr3up);
	\draw(botr)to[out=-90,in=-90](tr3bot);
	\draw(tr3bot)to[wavyup](tr3up);
	\draw(topm)to[out=90,in=90](tr2up);
	\draw(botm)to[out=-90,in=-90](tr2bot);
	\draw(tr2bot)to[wavyup](tr2up);
	\draw(out)to(bot);
	\draw(top)to(in);
}
+
\tikz[heighttwo,xscale=.5,baseline]{
	\coordinate(top)at(0,1.3){};	\coordinate(topr)at(1,1.3){};
	\coordinate(mid1)at(0,1){};	\coordinate(mid1r)at(1,1){};
	\coordinate(mid2)at(0,0.5){};	\coordinate(mid2r)at(1,0.5){};
	\coordinate(bot)at(0,0){};		\coordinate(botr)at(1,0){};
	\coordinate(toptr)at(1.5,1.3){};	\coordinate(bottr)at(1.5,0){};
	\coordinate(in)at(0,1.4){};
	\coordinate(out)at(0,-0.1){};
	\draw(bot)to(top);
	\draw(botr)to(topr);
	\draw(mid1)node[vector]{$A$};
	\draw(mid2)node[vector]{$B$};
	\draw(mid1r)node[vector]{$C$};
	\draw(bottr)to(toptr);
	\draw(topr)to[out=90,in=90](toptr);
	\draw(botr)to[out=-90,in=-90](bottr);
	\draw(top)to(in);
	\draw(out)to(bot);
}
+
\tikz[heighttwo,xscale=.5,baseline]{
	\coordinate(top)at(0,1.3){};	\coordinate(topr)at(1,1.3){};
	\coordinate(mid1)at(0,1){};	\coordinate(mid1r)at(1,1){};
	\coordinate(mid2)at(0,0.5){};	\coordinate(mid2r)at(1,0.5){};
	\coordinate(bot)at(0,0){};		\coordinate(botr)at(1,0){};
	\coordinate(toptr)at(1.5,1.3){};	\coordinate(bottr)at(1.5,0){};
	\coordinate(in)at(0,1.4){};
	\coordinate(out)at(0,-0.1){};
	\draw(bot)to(top);
	\draw(botr)to(topr);
	\draw(mid1)node[vector]{$A$};
	\draw(mid2)node[vector]{$C$};
	\draw(mid1r)node[vector]{$B$};
	\draw(bottr)to(toptr);
	\draw(topr)to[out=90,in=90](toptr);
	\draw(botr)to[out=-90,in=-90](bottr);
	\draw(top)to(in);
	\draw(out)to(bot);
}
+
\tikz[heighttwo,xscale=.5,baseline]{
	\coordinate(top)at(0,1.3){};	\coordinate(topr)at(1,1.3){};
	\coordinate(mid1)at(0,1){};	\coordinate(mid1r)at(1,1){};
	\coordinate(mid2)at(0,0.5){};	\coordinate(mid2r)at(1,0.5){};
	\coordinate(bot)at(0,0){};		\coordinate(botr)at(1,0){};
	\coordinate(toptr)at(1.5,1.3){};	\coordinate(bottr)at(1.5,0){};
	\coordinate(in)at(0,1.4){};
	\coordinate(out)at(0,-0.1){};
	\draw(bot)to(top);
	\draw(botr)to(topr);
	\draw(mid1)node[vector]{$A$};
	\draw(mid2r)node[vector]{$B$};
	\draw(mid1r)node[vector]{$C$};
	\draw(bottr)to(toptr);
	\draw(topr)to[out=90,in=90](toptr);
	\draw(botr)to[out=-90,in=-90](bottr);
	\draw(top)to(in);
	\draw(out)to(bot);
}
+
\tikz[heighttwo,xscale=.5,baseline]{
	\coordinate(top)at(0,1.3){};	\coordinate(topr)at(1,1.3){};
	\coordinate(mid1)at(0,1){};	\coordinate(mid1r)at(1,1){};
	\coordinate(mid2)at(0,0.5){};	\coordinate(mid2r)at(1,0.5){};
	\coordinate(bot)at(0,0){};		\coordinate(botr)at(1,0){};
	\coordinate(toptr)at(1.5,1.3){};	\coordinate(bottr)at(1.5,0){};
	\coordinate(in)at(0,1.2){};
	\coordinate(out)at(0,-0.3){};
	\draw(bot)to(top);
	\draw(top)to(in);
	\draw(out)to(bot);
	\draw(mid1)node[vector]{$A$};
	\draw(mid2)node[vector]{$B$};
	\draw(bot)node[vector]{$C$};
}
+
\tikz[heighttwo,xscale=.5,baseline]{
	\coordinate(top)at(0,1.3){};	\coordinate(topr)at(1,1.3){};
	\coordinate(mid1)at(0,1){};	\coordinate(mid1r)at(1,1){};
	\coordinate(mid2)at(0,0.5){};	\coordinate(mid2r)at(1,0.5){};
	\coordinate(bot)at(0,0){};		\coordinate(botr)at(1,0){};
	\coordinate(toptr)at(1.5,1.3){};	\coordinate(bottr)at(1.5,0){};
	\coordinate(in)at(0,1.2){};
	\coordinate(out)at(0,-0.3){};
	\draw(bot)to(top);
	\draw(top)to(in);
	\draw(out)to(bot);
	\draw(mid1)node[vector]{$A$};
	\draw(mid2)node[vector]{$C$};
	\draw(bot)node[vector]{$B$};
}
\right)	\\
&=& \frac{1}{6} \left(A \; \mathrm{tr}(B) \mathrm{tr}(C) + BA \; \mathrm{tr}(C) + CA \mathrm{tr}(B)
	+ A \; \mathrm{tr}(BC) + CBA + BCA \right)
\end{eqnarray*}
and we are done.
\end{proof}

\section{Problem Setup}

\subsection{Modelling the sampling process}

In the sampling process, we start by measuring the intensity of the signal:
\begin{equation}
y_0 = \| x \|_{\ell_2}^2 = \tr ( \Id X ).	\label{eq:intensity_measurement}
\end{equation}
This allows us to assume w.l.o.g. $\|x\|_{\ell_2} = 1$. 
Next, we
choose $m$ vectors $a_1, \ldots, a_m$ iid at random from a $t$-design $D_t \subset V^d$ and evaluate
\begin{equation}
 y_i = \tr (A_i X) =  | \langle x, a_i \rangle |^2	\quad \textrm{for } i=1, \ldots m,	\label{eq:amplitude_measurements}
\end{equation}
and consequently the vector $y = (y_1, \ldots, y_m )^T \in \mathbb{R}^m_+$ captures all the information we obtain from the sampling process. 
This process can be represented by a measurement operator 
\begin{eqnarray}
\AA: H^d & \to & \mathbb{R}^m,	\nonumber	\\	
Z & \mapsto& \sum_{i=1}^m  \tr( A_i Z) e_i	\label{eq:AA},
\end{eqnarray}
where $e_1, \ldots, e_m$ denotes the standard basis of $\mathbb{R}^m$.
Therefore $\AA (X) = y$ completely encodes the measurement process. 
For technical reasons we also consider the measurement operator
\begin{eqnarray}
	\RR: H^d &\to& H^d,	\nonumber \\	
	Z & \mapsto& m^{-1} \sum_{i=1}^m (d+1)d \; \Pi_{A_i} Z 
	= m^{-1} \sum_{i=1}^m (d+1)d\,A_i \tr( A_i Z), \label{eq:RR}
\end{eqnarray}
which is a renormalized version of $\AA^* \AA: H^d \to H^d$. Concretely
\begin{equation*}
\RR = \frac{(d+1)d}{m} \AA^* \AA .
\end{equation*}
The scaling  is going to greatly simplify our analysis, because it
guarantees 
that $\RR$ is ``near-isotropic'',
as the following result
shows.

\begin{lemma}[$\RR$ is near-isotropic]	\label{lem:isotropic}
The operator $\RR$ defined in (\ref{eq:RR}) is \emph{near-isotropic} in the sense that
\begin{equation}
\EE [ \RR] = \II + \Pi_{\Id} \quad \textrm{or} \quad
\EE \left[ \RR \right] Z = Z + \tr (Z) \Id \quad \forall Z \in H^d		\label{eq:near_isotropicity}
\end{equation}

\end{lemma}

\begin{proof}
Let us start with deriving (\ref{eq:near_isotropicity}). For $Z \in H^d$ arbitrary we have
\begin{eqnarray}
	\EE [ \RR] Z &=& \frac{(d+1)d}{m} \sum_{i=1}^m \EE [ A_i \tr (A_i Z) ] 	\nonumber \\
	&=& (d+1)d \tr_2 \left( \EE [ A_1^{\otimes 2} ] \Id \otimes Z \right) \label{eq:isoaux1}	\\
	&=& 2 \tr_2 \left( P_{\Sym^2} \Id \otimes Z \right) 	\label{eq:isoaux2}
	\\
	&=& Z + \Id (\tr Z) = \big(\mathcal{I}+\Pi_{\Id}\big) Z.  \nonumber
\end{eqnarray}
Here, (\ref{eq:isoaux1}) follows from the fact that the $a_i$'s are
chosen iid from a $t$-design,
(\ref{eq:isoaux2}) uses the fact that $\dim (\Sym^2) =
\binom{d+1}{2}^{-1}$ together with Definition \ref{def:design},
and the final line is  an application of 
Lemma \ref{lem:tr1sym2}.
\end{proof}

Let now $x \in V^d$ be the signal we want to recover. 
As in \cite{candes_phaselift_2012} we  consider the space
\begin{equation}
T := \left\{ xz^* + zx^*: \; z \in V^d \right\} \subset H^d
\end{equation}
(which is
the tangent space of the manifold of all hermitian matrices at the
point $X= xx^*$).  This space is of crucial importance for our
analysis. The orthogonal projection onto this space can be given
explicitly:
\begin{eqnarray}
\PP_T : H^d & \to & T,	\nonumber \\
Z & \mapsto& X Z + Z X - X Z X \label{eq:PPT1}\\
&=& X Z +Z X - (X,Z) X. \label{eq:PPT2}
\end{eqnarray}
We denote the projection onto its orthogonal complement with respect
to the Frobenius inner product
 by $\PP_T^\perp$. Then for any matrix $Z \in H^d$ the decomposition
\begin{equation*}
Z = \PP_T Z + \PP_T^\perp Z =: Z_T + Z_T^\perp
\end{equation*}
is valid. 
We point out that in particular
\begin{equation}
\PP_T \Pi_{\Id} \PP_T = \Pi_X 	\label{eq:aux1}
\end{equation}
holds. We will frequently use this fact. For a proof, consider $Z \in H^d$ arbitrary and insert the relevant definitions:
\begin{eqnarray*}
\PP_T \Pi_{\Id} \PP_T Z
&=& \PP_T \Id \tr (\Id \PP_T Z)
= \left( X \Id + \Id X - X \Id X \right) \tr \left( X Z + ZX - XZX \right) 	\\	
&=& X \tr(XZ)
= \Pi_X Z.
\end{eqnarray*}

\subsection{Convex Relaxation}

Following \cite{balan_signal_2006, candes_phaselift_2012, candes_solving_2012} the measurements (\ref{eq:intensity_measurement}) and (\ref{eq:amplitude_measurements}) can be translated into matrix form by applying the following ``lifts'':
\begin{equation*}
X :=  x x^*,	
\quad \textrm{and} \quad
A_i := a_i a_i^*.
\end{equation*}
By doing so the measurements assume the a linear form:
\begin{eqnarray*}
y_0 &=& \| x \|_2^2 = (\Id, X) = \tr (X), 	\\
y_i &=& \left(A_i,X\right)= \Tr \left(A_i X \right) \quad i=1,\ldots,m.
\end{eqnarray*}
Hence, the phase retrivial problem becomes a matrix recovery problem.
The solution to this is guaranteed to have rank 1 and encodes (up to a
global phase) the unknown vector $x$ via $X = x x^*$. Relaxing the
rank minimization problem (which would output the correct solution)
to a trace norm minimization yields the now-familiar convex
optimization problem

\begin{eqnarray}
\textrm{minarg}_{X'} & &  \| X' \|_1	\label{eq:convex_program}\\
\textrm{subject to} & & \left( A_i, X' \right) = y_i \quad i = 1,\ldots m,	\nonumber \\
& & X' = \left(X' \right)^\dagger,	\nonumber	\\ 
& & \tr (X') = 1,	\nonumber \\
& & X' \geq 0. \nonumber
\end{eqnarray}
While this convex program is formally equivalent to the previously
studied general-purpose
matrix recovery algorithms \cite{recht_guaranteed_2010,
candes_power_2010, gross_recovering_2011}, there are two important 
differences:
\begin{itemize}
	\item The measurement matrices $A_i$ are rank-1 projectors: $A_i = a_i
	a_i^*$. 
	\item 
	The unknown signal is known to be proportional to a rank-1
	projector ($X = x x^*$) as well.
\end{itemize}
While the second fact is clearly of advantage for us, the first one
makes the problem considerably harder: In the language of
\cite{gross_recovering_2011}, 
it means that the ``incoherence
parameter'' $\mu = d \max_{i=1,\ldots,m} \| A_i \|_\infty = d \|a_i
\|_{\ell_2}^2 =d$ is as large as it can get! Higher values of $\mu$
correspond to more ill-posed problems and as a result, a direct
application of previous low-rank matrix recovery results fails.  It is
this problem that Refs.~\cite{candes_phaselift_2012,
candes_solving_2012} first showed how to circumvent for the case of
Gaussian measurements. Below, we will adapt these ideas to the case of
measurements drawn from designs, which necessitates following more closely
the approach of 
\cite{gross_recovering_2011}.

\subsection{Well-posedness / Injectivity}

In this section, we follow
\cite{candes_phaselift_2012,gross_recovering_2011} to establish a
certain injectivity property of the measurement operator $\AA$.
Compared to \cite{candes_phaselift_2012}, our injectivity properties
are somewhat weaker. Their proof used the independence of the
components of the Gaussian measurement operator, which is not
available in this setting, where individual vector components might be
strongly correlated. We will pay the price for these weaker
bounds in Section~\ref{sec:dual}. There, we construct an 
``approximate dual certificate'' 
that proves that the sought-for signal indeed 
minimizes the nuclear norm. Owing to the weaker bounds found here, the
construction is more complicated than in \cite{candes_phaselift_2012}.
In the language of \cite{gross_recovering_2011}, we will have to carry
out the full ``golfing scheme'', as opposed to the ``single leg'' that
proved sufficient in \cite{candes_phaselift_2012}.

\begin{proposition}		\label{prop:inj1}
	With probability of failure smaller than $d^2 \exp (- \frac{3 m}{384
	d})$ the inequality
	\begin{equation}
		0.25 d^{-2} \| Z \|_2^2 < m^{-1} \| \AA(Z) \|_2^2		\label{eq:inj1}
	\end{equation}
	is valid for all matrices $Z \in T$ simultaneously. 
\end{proposition}

\begin{proof}
	We aim to show the more general statement
	\begin{equation*}
		\Pr \left[ m^{-1} \| \AA (Z) \|_2^2 < 0.5 (1-\delta) \|Z \|_2^2 \; \; \forall Z \in T \right] 
		\leq d^2 \exp \left( -\frac{3m\delta^2}{96d} \right)
	\end{equation*}
	for any $\delta \in (0,1)$. 
	
	For $Z \in T$ abritrary use near-isotropicity of $\RR$ ($\EE[ \RR] =
	\II + \Pi_\Id$) and observe
	\begin{eqnarray}
		 & & m^{-1} \| \AA (Z) \|_2^2 		\nonumber \\
		&=& m^{-1}  \sum_{i=1}^m \left( \tr (Z A_i) \right)^2	
		= \tr ( Z m^{-1} \sum_i A_i \tr (A_i Z) )	
		= \frac{1}{(d+1)d} \tr (Z \RR Z)  	\nonumber	\\
		&=& \frac{1}{(d+1)d} \tr( Z (\RR-\EE[\RR] ) Z) + \frac{1}{(d+1)d} \tr(Z(\II + \Pi_\Id) Z)	\nonumber	\\
		&=& \frac{1}{(d+1)d}\tr(Z\PP_T (\RR-\EE[\RR])\PP_T Z) +
		\frac{1}{(d+1)d}(\tr (Z^2) + (\tr Z)^2)	\nonumber	\\
		& \geq & 0.5 d^{-2} \left( \tr(Z \PP_T (\RR - \EE[\RR])\PP_T Z) + \tr(Z^2) \right)	\nonumber	\\
		& \geq & 0.5 d^{-2} (1 + \lambda_{\min} \left( \PP_T (\RR - \EE[\RR]) \PP_T \right) \| Z \|_2^2 ,	\label{eq:inj1aux1}
	\end{eqnarray} 
	where we have used $\PP_T Z = Z$ as well as $\MM \geq
	\lambda_{\min}(\MM)\II$ for any operator $\MM$.
	Therefore everything boils down to bounding the smallest eigenvalue
	of $\PP_T(\RR-\EE[\RR]) \PP_T$.
	To this end we aim to apply Theorem
	\ref{thm:smallest_eigenvalue_bernstein} and decompose 
	\begin{equation*}
		\PP_T(\RR-\EE[\RR]) \PP_T
		= \sum_{i=1}^m \left( \MM_i - \EE[ \MM_i ] \right) 
		\quad \textrm{with} \quad
		\MM_i = \frac{(d+1)d}{m} \PP_T \Pi_{A_i} \PP_T.
	\end{equation*}
	Note that these summands have mean zero by construction. Furthermore
	observe that the auxiliary result (\ref{eq:aux1}) 
	implies
	\begin{eqnarray*}
	- \frac{2}{m} \II
	&\leq& - \frac{1}{m}\II - \frac{1}{m} \Pi_X
	\leq - \frac{1}{m} \PP_T \II \PP_T - \frac{1}{m} \PP_T \Pi_{\Id} \PP_T 	\\
	&=& - \PP_T \EE[\MM_i] \PP_T 
	\leq \PP_T (\MM_i - \EE[\MM_i] ) \PP_T 
	\end{eqnarray*}
	and the a priori bound 
	\begin{equation*}
	\lambda_{\min} (\MM_i - \EE[ \MM_i] ) \geq -2/m =: - \underline{R}
	\end{equation*}
	follows.
	For the variance we use the standard identity
	\begin{equation*}
	0 \leq \EE[ (\MM_i - \EE[\MM_i] )^2] = \EE[\MM_i^2] - \EE[\MM_i]^2 \leq \EE[\MM_i^2] 
	\end{equation*}
	and focus on the last expression. Writing it out explicitly yields
	\begin{eqnarray*}
	0 \leq \EE [ \MM_i^2 ]
	&=& \frac{(d+1)^2 d^2}{m^2} \PP_T \EE \left[ \Pi_{A_i} \PP_T \Pi_{A_i} \right] \PP_T 		\\
	&=& \frac{(d+1)^2 d^2}{m^2} \PP_T \EE \left[ \tr(A_i \PP_T A_i) \Pi_{A_i} \right] \PP_T .
	\end{eqnarray*}
	The trace can be bounded from above by 
	\begin{eqnarray*}
	\tr(A_i \PP_T A_i ) 
	&=& \tr \big( A_i ( XA_i + A_i X - \tr(A_iX)X) \big)	\\
	&=& 2 \tr ( A_i X ) - \tr (A_i X)^2 \leq 2 \tr (A_i X),
	\end{eqnarray*}
	where we have used the basic definition of $\PP_T$ and $0 \leq \tr (A_i X) = | \langle a_i, x \rangle |^2 \leq 1$. 
	Consequently, for $Z \in T$ arbitrary 
	\begin{eqnarray*}
	&& \PP_T \EE[ \MM_i ^2 ] \PP_T Z \\
	&\leq& \frac{2(d+1)^2 d^2}{m^2} \PP_T \EE \left[ A_i \tr (A_i X) \tr (A_i Z) \right]	\\
	&=& \frac{2(d+1)^2 d^2}{m^2} \PP_T \tr_{2,3} \left( \EE[ A_i^{\otimes 3}] \Id \otimes X \otimes Z \right)		\\
	&=& \frac{12(d+1)^2 d^2}{m^2 (d+2)(d+1)d} \PP_T \tr_{2,3} \left( P_{\Sym^3} \Id \otimes X \otimes Z \right) 	\\
	& \leq & \frac{2d}{m^2} \PP_T \left( \Id \tr(Z) + X \tr(Z) + Z + \Id \tr(XZ) + ZX + XZ \right)	\\
	& = & \frac{2d}{m^2} \left( X \tr (XZ) + X \tr(XZ) +  Z + X \tr (XZ) + \PP_T Z + X \tr(XZ) \right) 	\\
	&=& \frac{2d}{m^2} \left( 4 \Pi_X + 2 \II \right) Z 	
	 \leq  \frac{12d}{m^2} \II  Z.
	\end{eqnarray*}
	Here we have applied $\dim \Sym^3 = \binom{d+2}{3}^{-1}$ and Lemma \ref{lem:P3_calculation}
	in lines 3 and 4, respectively. 
	Furthermore we used $Z \in T$ -- hence $\PP_T Z = Z$ and $\tr (Z) = \tr (XZ)$ -- 
	as well as the basic definition (\ref{eq:PPT2}) of $\PP_T$ 
	to simplify the terms occuring in the fourth line. 
	Putting everything together yields
	\begin{equation*}
	\EE [ (\MM_i - \EE[\MM_i])^2 ] \leq \EE[ \MM_i^2] \leq \frac{12d}{m^2} \II
	\end{equation*}
	and we can safely set $\sigma^2 := \frac{12d}{m}$. Now Theorem \ref{thm:smallest_eigenvalue_bernstein} tells us
	\begin{equation*}
	\Pr \left[ \lambda_{\min} \left( \PP_T (\RR - \EE[\RR]) \PP_T \right) \leq - \delta \right] \leq d^2 \exp \left( - \frac{3m\delta^2}{8 \times 12d} \right)
	\end{equation*}
	for all $0 \leq \delta \leq 1 \leq 6d = \sigma^2/\underline{R}$. 
	This gives the desired bound on the event 
	\begin{equation*}
		\left\{\lambda_{\min} (\PP_T (\RR - \EE[\RR] )\PP_T) \leq - \delta \right\}
	\end{equation*}
	occuring. If this is not the case,
	 (\ref{eq:inj1aux1}) implies
	\begin{equation*}
	m^{-1} \| \AA (Z) \|_{\ell_2}^2 > 0.5 d^{-2}( 1- \delta) \|Z \|_2^2 
	\end{equation*}
	for all matrices $Z \in T$ simultaneously. This is the general statement at the beginning of the proof
	and setting $\delta = 1/2$ yields Proposition \ref{prop:inj1}.
\end{proof}

\begin{proposition}		\label{prop:inj2}
Let $\AA$ be as above with vectors sampled from a $t$-design ($t \geq 1$). Then the statement
\begin{equation}
m^{-1} \| \AA (Z) \|_{\ell_2}^2 \leq   \| Z \|_2^2	\label{eq:inj2}
\end{equation}
holds with probability one 
for all matrices $Z \in H^d$ simultaneously. 
\end{proposition}

\begin{proof}
Pick $Z \in H^d$ arbitrary and observe
\begin{eqnarray*}
\| \AA (Z) \|_{\ell_2}^2 = \frac{1}{m} \sum_{i=1}^m  \left( \tr (A_i Z) \right)^2 
= \tr \left( Z \left( \frac{1}{m} \sum_{i=1}^m \Pi_{A_i}  \right)Z \right)
\leq  \tr (Z \II Z) =  \| Z\|_2^2, 
\end{eqnarray*}
where we  have used $0 \leq \Pi_{A_i} \leq \II$. 
\end{proof}

Note that equation (\ref{eq:inj2}) can be improved. Indeed, a standard
application of the Operator Bernstein inequality (Theorem
\ref{thm:bernstein}) 
gives
\begin{equation*}
m^{-1} \| \AA (Z) \|_{\ell_2}^2 \leq 2d^{-1} \| Z \|_2^2
\end{equation*}
for all matrices $Z \in T$ with probability of failure smaller than $d^2 \exp \left( - Cm/d \right)$ for some $ 0 < C \leq 1$. 
However, we actually do not require this tighter bound.

\section{Proof of the Main Theorem / Convex Geometry}	\label{sec:convex_geometry}

In this section, we will follow
\cite{gross_recovering_2011, candes_power_2010} to prove that the
convex program
(\ref{eq:convex_program}) indeed recovers the sought for signal $x$,
provided that a certain geometric object -- an
\emph{approximate dual certificate} -- exists.

\begin{definition}[Approximate dual certificate]	\label{def:dual_cert}
Assume that the sampling process corresponds to (\ref{eq:intensity_measurement}) and (\ref{eq:amplitude_measurements}).
Then we call $Y \in H^d$ an \emph{approximate dual certificate}, provided that
$Y \in \mathrm{span} \left( \Id, A_1 ,\ldots, A_m \right)$ and
\begin{equation}
\| Y_T - X \|_2 \leq \frac{1}{4d}
\quad \textrm{as well as} \quad
\| Y_T^\perp \|_\infty \leq \frac{1}{2}.	\label{eq:Y}
\end{equation}
\end{definition}

\begin{proposition}	\label{prop:convex geometry}
Suppose that the measurement gives us access to $ \|x \|_{\ell_2}^2$ and $y_i = | \langle a_i, x \rangle |^2$ for $i=1,\ldots, m$. 
Then the convex optimization (\ref{eq:convex_program}) recovers the unknown $x$ (up to a global phase) provided that (\ref{eq:inj1}) holds and an approximate dual certificate $Y$ exists.
\end{proposition}

\begin{proof}
Let $\tilde{X} \in H^d$ be an arbitrary feasible point of (\ref{eq:convex_program}) and
decompose it as $ \tilde{X} = X + \Delta$. Feasibility then implies $\AA (\tilde{X}) = \AA(X)$
and  $\AA (\Delta) = 0$ must in turn hold for any feasible displacement $\Delta $.
Now the pinching inequality \cite{bhatia_matrix_1997}  (Problem II.5.4) implies
\begin{equation*}
\| \tilde{X} \|_1 = \| X + \Delta \|_1 
\geq \| X \|_1 + \tr (\Delta_T) + \| \Delta_T^\perp \|_1.
\end{equation*}
Consequently $X$   is guaranteed to be the unique minimum of (\ref{eq:convex_program}), if
\begin{equation}
\tr (\Delta_T) + \| \Delta_T^\perp \|_1 > 0 	\label{eq:optimality}
\end{equation}
is true for every feasible $\Delta$. 
In order to show this we combine feasibility of $\Delta$ with inequalities (\ref{eq:inj1}) and (\ref{eq:inj2}) to obtain
\begin{equation}
\| \Delta_T \|_2 < 2d m^{-1/2} \| \AA (\Delta_T) \|_{\ell_2} = 2d m^{-1/2} \| \AA (\Delta_T^\perp ) \|_{\ell_2}
\leq 2d \| \Delta_T^\perp \|_2.	\label{eq:inj3}
\end{equation}
Feasibility of $\Delta$ also implies $(Y,\Delta) = 0$, because by defnition $Y$  is in the range of $\AA^*$. Combining this insight with the defining property  (\ref{eq:Y})  of $Y$ and (\ref{eq:inj3}) 
yields
\begin{eqnarray*}
0 &=&
(Y, \Delta) = (Y_T - X , \Delta_T ) + ( X, \Delta_T ) + (Y_T^\perp, \Delta_T^\perp )	
 \\
& \leq & \| Y_T - X \|_2 \| \Delta_T \|_2 + \tr ( \Delta_T) + \|Y_T^\perp \|_\infty \| \Delta_T^\perp \|_1	\\
&<& \tr(\Delta_T) + \|Y_T - X \|_2 2d\| \Delta_T^\perp \|_2 + \| Y_T^\perp \|_\infty \|\Delta_T^\perp \|_1	r \\
& \leq & \tr (\Delta_T) + 1/2 \| \Delta_T^\perp \|_2 + 1/2 \| \Delta_T^\perp \|_1	 \\
& \leq & \tr (\Delta_T) + \| \Delta_T^\perp \|_1 ,
\end{eqnarray*}
which is just the desired optimality criterion (\ref{eq:optimality}).
\end{proof}

\section{Constructing the Dual Certificate}\label{sec:dual}

A straightforward approach to construct an approximate dual
certificate would be to set
\begin{equation}
Y =  \RR X - \tr (X) \Id = \frac{(d+1)d}{m} \sum_{i=1}^m A_i \tr (A_i X) - \tr (X) \Id \in 
\mathrm{span} \left( \Id, A_1, \ldots, A_m \right).	\label{eq:ansatz}
\end{equation}
In expectation, $\EE[Y]=X$, which is the ``perfect dual certificate''
in the sense that the norm bounds in (\ref{eq:Y}) vanish.
The hope would be to use the Operator Bernstein inequality to show that
with high probablity, $Y$ will be sufficiently close to its
expectation. 
It has been shown that a slight refinement of the ansatz
(\ref{eq:ansatz}) indeed achieves this
goal
Ref.~\cite{gross_recovering_2011, kueng_ripless_2013}. However, the
Bernstein bounds depend on the worst-case operator norm of the
summands. In our case, they can be as large as $d^2|\langle a_i,
x\rangle|^2$, which can reach $d^2$. This is far larger than in previous
low-rank matrix recovery problems. Ref.~\cite{candes_phaselift_2012}
relied on the fact that large overlaps $|\langle a_i,
x\rangle|^2 \gg \mathcal{O}(d^{-1})$ are
``rare'' for Gaussian $a_i$.

The key observation here is that the $t$-design property provides
one with useful information about the first $t$ moments of the random
variable
$| \langle x, a_i \rangle |^2$.  This
knowledge allows
us to explicitly bound the probability of ``dangerously large
overlaps'' or ``coherent measurement vectors'' occurring.

\begin{lemma}[Undesired events] \label{lem:undesired}
Let $x \in V^d$ be an arbitrary vector of unit length. If $a$ is chosen uniformly at random from a $t$-design ($t \geq 1$) $D_t \subset V^d$, then the following is true for every $\gamma \leq 1$:
\begin{equation}
\Pr  \left[ | \langle a, x \rangle |^2 \geq 5 t d^{-\gamma} \right]  \leq 4^{-t} d^{-t(1-\gamma)}.	\label{eq:undesired}
\end{equation}
\end{lemma}

\begin{proof}
We aim to prove the slightly more general statement
\begin{equation*}
\Pr \left[ | \langle a, x \rangle |^2 \geq  (\delta +1) t d^{-\gamma} \right] \leq \delta^{-t} d^{-t(1-\gamma)},
\end{equation*}
which is valid for any $\delta \geq 1$. Setting $\delta = 4$ then yields (\ref{eq:undesired}). 
The $t$-design property provides us with useful information about the first $t$ moments of the non-negative random variable 
$ \xi = | \langle a, x \rangle |^2 $. 
Indeed, with $A = a a^*$ it holds for every $k \leq t$ that 
\begin{eqnarray*}
\EE \left[ \xi^k \right]
&=& \EE \left[ \tr (AX)^k \right]	\\
&=& \tr \left( \EE \left[ A^{\otimes k} \right] X^{\otimes k} \right)	\\
&=& \binom{d+k-1}{k}^{-1} \tr \left( P_{\Sym^k} X^{\otimes k} \right)	\\
&=& \binom{d+k-1}{k}^{-1} \tr \left( X^{\otimes k} \right)	\\
& \leq& d^{-k} k!,
\end{eqnarray*}
because $X^{\otimes k}$ is invariant under $P_{\Sym^k}$.
One way of seing this\footnote{Alternatively one could also rearange tensor systems: $X^{\otimes k} = (x x^*)^{\otimes k} \simeq x^{\otimes k} (x^*)^{\otimes k}$ and use $P_{\Sym^k} x^{\otimes k} = x^{\otimes k}$.} 
is to note that $\range(X^{\otimes k}) = \mathrm{span}(x^{\otimes k})$ and the latter is already contained in $\Sym^k$.
Therefore the $k$-th moment $\tau_k$ of $\xi$ is bounded by
\begin{equation*}
\tau _k = \left ( \EE[ \xi^k ] \right)^{1/k} \leq (d^{-k} k!)^{1/k} \leq k /d.
\end{equation*}
These inequalities are tight for the mean $\mu = \tau_1$ of $\xi$ and hence
\begin{equation*}
\mu = \EE [\xi] = d^{-1}.
\end{equation*}
Now we aim to use the well-known $t$-th moment bound
\begin{equation*}
\Pr \left[ | \xi - \mu | \geq s \tau _t \right] \leq s^{-t},
\end{equation*}
which is a straightforward generalization of Chebyshev's inequality. Applying it, yields the desired result. Indeed,
\begin{eqnarray*}
\Pr \left[ | \langle a, x \rangle |^2 \geq (\delta + 1) t d^{-\gamma} \right] 
&=& \Pr \left[ \xi- \mu \geq (\delta+1) t d^{-\gamma} - d^{-1} \right]	\\
& \leq & \Pr \left[ \xi - \mu \geq \delta t d^{-\gamma} \right]	\\
& \leq & \Pr \left[ | \xi - \mu | \geq \delta d^{1-\gamma} \tau_t \right]	\\
& \leq & \delta^{-t} d^{-t(1-\gamma)},
\end{eqnarray*}
and we are done. 
\end{proof}

The previous lemma bounds the probability of the undesired events
\begin{equation}
E_i ^c = \left\{ | \langle a_i, x \rangle |^2 \geq 5 t d^{-\gamma} \right\},
\end{equation}
where $0 \leq \gamma \leq 1$ is a fixed parameter which we refer to as the \emph{truncation rate}.
It turns out that a single truncation of this kind does not quite suffice yet for our purpose.
We need to introduce a second truncation step.

\begin{definition}\label{def:postselect}
Fix $Z \in T$ arbitrary and decompose it as 
\begin{equation*}
Z = \zeta \left( x z^* + z x^* \right),
\end{equation*}
for some unique $\zeta >0 $ and $z \in V^d$ with $\| z \|_{\ell_2} = 1$. For this $z$ we introduce the event
\begin{equation*}
G_i^c := \left\{ | \langle z, a_i \rangle |^2 \geq 5 t d^{-\gamma} \right\}
\end{equation*}
and define the two-fold truncated operator
\begin{equation}
\RR_Z := \RR_z = \frac{(d+1)d}{m} \sum_{i=1}^m 1_{E_i} 1_{G_i} \Pi_{A_i},	\label{eq:RRz}
\end{equation}
where $1_{E_i}$ and $1_{G_i}$ denote the indicator functions associated with the events $E_i$ and $G_i$, respectively. 
\end{definition}

The following result shows that due to Lemma \ref{lem:undesired} this truncated operator is in expectation close to the original $\RR$.

\begin{proposition}		\label{prop:postselect}
Fix $Z \in T$ arbitrary and let $\RR_Z$ be as in (\ref{eq:RRz}). Then
\begin{equation}
\| \EE[ \RR_Z - \RR ] \|_{\op} \leq 4^{1-t} d^{2-t(1-\gamma)} 
\end{equation}
\end{proposition}

\begin{proof}

We start by introducing the auxiliar (singly truncated) operator
\begin{equation*}
\RR_{\mathrm{aux}} := \frac{(d+1)d}{m} \sum_{i=1}^m 1_{E_i} \Pi_{A_i}	
\end{equation*}
and observe
\begin{equation}
\| \EE \left[ \RR_Z - \RR \right] \|_{\op}
\leq \| \EE \left[ \RR - \RR_{\mathrm{aux}} \right] \|_{\op}
+ \| \EE \left[ \RR_Z - \RR_{\mathrm{aux}} \right] \|_{\op}.	\label{eq:postselect1}
\end{equation}
Now use Lemma \ref{lem:undesired} to bound the first term:
\begin{eqnarray*}
\left\| \EE[ \RR - \RR_{\mathrm{aux}} ] \right\|_\op
&=& \left\|\frac{(d+1)d}{m} \sum_{i=1}^m \EE \left[ (1-1_{E_i})\Pi_{A_i} \right] \right\|_\op 	\\		
& \leq & \frac{(d+1)d}{m} \sum_{i=1}^m  \EE \left[ 1_{E_i^c} \| \Pi_{A_i} \|_\op \right] 		\\
 &\leq&  \frac{2d^2}{m} \sum_{i=1}^m \EE \left[ 1_{E_i^c} \right]	
= \frac{2d^2}{m} \sum_{i=1}^m \Pr [ E_i^c] 									\\
&\leq& 2d^2 \times 4^{-t} d^{-t(1-\gamma)} 	
= 2^{1-2t} d^{2-t(1-\gamma)}.
\end{eqnarray*}
Similarily,
\begin{eqnarray*}
\left\| \EE \left[ \RR_{\mathrm{aux}} - \RR_Z \right] \right\|_\op
&=& \frac{(d+1)d}{m} \left\| \sum_{i=1}^m \EE \left[ 1_{G_i^c} \Pi_{A_i} \right] \right\|_\op
\leq \frac{2d^2}{m} \sum_{i=1}^m \EE[ 1_{G_i^c} ] 		\\
& \leq & \frac{2d^2}{m} \sum_{i=1}^m \Pr [ G_i^c ] 
\leq 2^{1-2t} d^{2-t(1-\gamma)}
\end{eqnarray*}
and inserting these bounds into (\ref{eq:postselect1}) yields the desired statement. 
\end{proof}

We now establish a technical result which will allow us to find a
suitable approximate dual certificate using the ``golfing scheme''
construction
\cite{gross_recovering_2011, kueng_ripless_2013}.

\begin{proposition}		\label{prop:golfaux}
Fix $Z \in T$ arbitrary, let $\RR_Z$ be as in (\ref{eq:RRz}). Assume that that the design order $t$ is at least 3 and the truncation rate $\gamma$ satisfies 
\begin{equation*}
\gamma \leq 1 - 2/t.
\end{equation*}
Then for $1/4 \leq b \leq 1$ and $c \geq \sqrt{2}b$
 with probability at least
$ 1- d \exp (-\frac{9mb}{640 t d^{2-\gamma}})$ 
one has

\begin{eqnarray}
\| \PP_T^\perp \left( \RR_Z Z - \tr (Z) \Id \right) \|_\infty &\leq& b \| Z \|_2 \quad \textrm{and} \label{eq:golf1} \\
\| \PP_T \left( \RR_Z - Z - \tr (Z) \Id \right) \|_2 & \leq & c \| Z \|_2.\label{eq:golf2}
\end{eqnarray}

\end{proposition}

\begin{proof}
The statement is invariant under rescaling of $Z$. Therefore it suffices to treat the case $\|Z\|_2 = 1$.
In this case we can decompose 
\begin{equation*}
Z =  \zeta (zx^* + x z ^* )
\end{equation*}
with some fixed $z \in V^d$ obeying $\|z\|_{\ell_2} = 1$ and $0 < \zeta \leq 1$. 
Near-Isotropicity (Lemma \ref{lem:isotropic}) of $ \RR$ guarantees
$ \PP_T^\perp  \EE [ \RR] Z =\tr (Z) \PP_T^\perp Z$ as well as $\PP_T \EE[\RR] Z = Z + \tr (Z) \PP_T \Id$. 
Let us now focus on (\ref{eq:golf1}) and use Proposition \ref{prop:postselect}
in order to write
\begin{eqnarray*}
	&& \| \PP_T^\perp \left( \RR_Z Z - \tr (Z) \Id \right) \|_\infty \\
&=& \| \PP_T^\perp  \left( \RR_Z - \EE[\RR] \right) Z \|_\infty	\\
& \leq & \| \PP_T^\perp  \left( \RR_Z - \EE[\RR_Z] \right) Z \|_\infty
+ \| \PP_T^\perp  \EE[\RR_Z-\RR ] Z \|_\infty 	\\
& \leq & \| \PP_T^\perp \|_\op \| (\RR_Z - \EE[\RR_Z])Z \|_\infty
+ 4^{1-t} d^{2-t(1-\gamma)} \| \PP_T^\perp \|_\op  \| Z \|_2 \\
& \leq & \| (\RR_Z - \EE [ \RR_Z])Z \|_\infty + b/4.
\end{eqnarray*}
Here we have used  $ \| \PP_T^\perp \|_\op \leq 1$ as well as
\begin{equation}
\| \EE \left[ \RR_Z - \RR \right] \|_\op \leq 4^{1-t} d^{2-t(1-\gamma)} \leq 4^{1-t} \leq 1/16 \leq b/4,		\label{eq:golfaux1}
\end{equation}
which follows from $\gamma \leq 1- 2/t$, $t \geq 3$ and $b \geq 1/4$. 
To obtain (\ref{eq:golf2}) we use a similar reasoning:
\begin{eqnarray*}
&&
 \| \PP_T \left( \RR_Z Z- Z - \tr (Z) \Id \right)  \|_2 	 \\
& = & \| \PP_T  \left( \RR_Z - \EE[\RR] \right) Z \|_2	\\
&\leq & \sqrt{2} \| \PP_T  \left( \RR_Z - \EE[\RR_Z] \right) Z \|_\infty
+ \| \PP_T \EE[\RR_Z - \RR] Z \|_2 		\\
& \leq & \sqrt{2} \| \PP_T \|_\op \| (\RR_Z - \EE[\RR_Z] ) Z \|_\infty 
+ b/4 \| \PP_T \|_\op  \| Z \|_2 \\
& \leq & \sqrt{2} \| \left( \RR_Z - \EE[\RR_Z] \right) Z \|_\infty + b/4,
\end{eqnarray*}
where we have used the fact that $\PP_T$ projects onto a subspace of at most rank-2 matrices in the third line and
(\ref{eq:golfaux1}) in the fourth.
This motivates to define the event
\begin{equation*}
E := \left\{ \| \left( \RR_Z - \EE[ \RR_Z] \right) Z \|_\infty \leq 3b/4 \right\}
\end{equation*}
which guarantees both (\ref{eq:golf1}) and (\ref{eq:golf2})
due to the assumption on $c$ and $ \| Z \|_2 = 1$.
So everything boils down to bounding the probability of $E^c$. 
We decompose
\begin{equation*}
\left( \RR_Z - \EE[\RR_Z] \right) Z = \sum_{i=1}^m \left( M_i - \EE[M_i] \right)
\quad \textrm{with} \quad
M_i = \frac{(d+1)d}{m} 1_{E_i} 1_{G_i} A_i \tr (A_i Z).
\end{equation*}
We will estimate this sum using the Operator Bernstein inequality (Theorem \ref{thm:bernstein}).
Thus we need an a priori bound for the summands
\begin{eqnarray*}
\| M_i \|_\infty
&=& \frac{(d+1)d}{m} 1_{E_i} 1_{G_i} \| A_i \|_\infty | \tr (A_i Z ) | 
\leq \frac{2 d^2}{m} 1_{E_i} 1_{G_i} 2 | \langle x, a_i \rangle | | \langle z, a_i \rangle | \\
& \leq & \frac{4d^2}{m}5td^{-\gamma} = \frac{20}{m} t d^{2-\gamma} =: \overline{R},
\end{eqnarray*}
as well as a bound for the variance. 
First observe that
\begin{equation*}
\EE [ (M_i - \EE [ M_i] )^2 ]
= \EE \left[ M_i ^2 \right] - \EE[M_i]^2 \leq \EE \left[ M_i^2 \right] .
\end{equation*}
and therefore 
\begin{eqnarray*}
&&\EE \left[ M_i^2 \right] \\
&= & \frac{(d+1)^2 d^2}{m^2} \EE \left[ 1_{E_i} 1_{G_i} \tr (A_i Z)^2 A_i^2 \right]	
 \leq  \frac{(d+1)^2 d^2}{m^2} \EE \left[ \tr (A_i Z)^2 A_i^2 \right]	\\
&=& \frac{(d+1)^2 d^2}{m^2} \tr_{2,3} \left( \EE[ A_i^{\otimes 3} ] \Id \otimes Z \otimes Z \right) 	
= \frac{6(d+1)d}{m^2(d+2)} \tr_{2,3} \left( P_{\Sym^3} \Id \otimes Z \otimes Z \right) \\
&\leq& \frac{d}{m^2} \left( \Id \tr (Z)^2 + Z \tr (Z) + Z + \Id \tr(Z^2) + 2Z^2 \right) 	\\
& \leq & \frac{8d}{m^2} \| Z \|_2^2 \Id = \frac{8d}{m^2} \Id.
\end{eqnarray*}
Here we have used $\tr(Z) \leq \sqrt{2} \| Z \|_2$, $ Z^2 \leq \| Z \|_2^2 \Id$ and $ \|Z \|_2 = 1$. 
From this we can conclude
\begin{equation*}
\Big\| \sum_i \EE [ ( M_i - \EE[M_i] )^2 \Big\|_\infty 
\leq m \max_{i=1,\ldots,m} \| \EE[M_i^2] \|_\infty 
\leq \frac{8d}{m} 
=: \sigma^2.
\end{equation*}
Observing that 
\begin{equation*}
\frac{\sigma^2}{\overline{R}} \leq \frac{8}{20t}d^{\gamma-1} \leq \frac{2}{15} \leq \frac{3}{4}b,
\end{equation*}
Theorem \ref{thm:bernstein} yields
\begin{equation*}
\Pr \left[ E^c \right] =  \Pr \left[ \| \left(\RR_Z - \EE[ \RR_Z] \right) Z \|_\infty > 3b/4 \right] 
\leq d \exp \left( -\frac{3 \times 3 mb}{8 \times 4 \times 20 t d^{2-\gamma}} \right),
\end{equation*} 
as desired. 
\end{proof}

With this ingredient we can now construct a suitable approximate dual certificate $Y$, closely following \cite{kueng_ripless_2013}.

\begin{proposition}		\label{prop:Y}
Let $x \in V^d$ be an arbitrary normalized vector ($\| x\|_{\ell_2} =1$), $X = x x^*$ and let $\omega \geq 1$ be arbitrary.
If the design order $t$ ($t \geq 3$) and the truncation rate $\gamma$
is chosen such that
\begin{equation*}
\gamma \leq 1 - 2/t
\end{equation*}
holds and the total number of measurements fulfills
\begin{equation}\label{eqn:measurements}
m \geq C \omega t d^{2-\gamma} \log ^2 (d),
\end{equation}
then with probability larger than $1 - 0.5 \mathrm{e}^{-\omega}$,
there exists an approximate dual certificate $Y$ as in
Def.~\ref{def:dual_cert}. Here, $C$ is a universal constant 
(which can in principle be recoverd explicitly from the proof). 
\end{proposition}

\RestyleAlgo{boxruled}
\begin{algorithm}\label{alg:dual}
\caption{
	Summary of the randomized ``golfing scheme''
	\cite{gross_recovering_2011} used in the proof of Prop.~\ref{prop:Y}
	to show the existence of an approximate dual
	certificate\protect\footnotemark.
} 
\textbf{Input}:\newline
\begin{tabular}{lll}
	\qquad & $X \in H^d$ \quad &  $\#$ signal to be recovered \\
				 & $l \in \mathbb{N}$ & $\#$ maximum number of iterations\\
				 & $\left\{ m_i \right\}_{i=1}^l \subset   \mathbb{N}$\quad&
	$\#$ number of measurement vectors used in $i$th iteration \\
	& $r$
	& $\#$   require $r$ ``successful'' iterations \\
	&& $\#$ (i.e.\ iterations where we enter the inner
	\textbf{if}-block)
\end{tabular}

\ \\
 
\textbf{Initialize:}\newline
\begin{tabular}{lll}
	\quad & $\mathbf{Y} = [\,]$ \quad & $\#$ a list of 
	matrices in $H^d$, initially empty \\
	& $\mathbf{Q} = [X]$ \quad & $\#$ a list of matrices in $T$,
	initialized to hold $X$ as its only element \\
	& $i = 1$ & $\#$ number of current iteration \\
	& $\xi = [0,\ldots,0]$ & $\#$ array of $l$ zeros; $\xi_i$ will be set to $1$
	if $i$th iteration succeeds 		
\end{tabular}

\ \\

\textbf{Body:}\newline
\While{ $i \leq l$ and $\sum_{j=1}^i \xi_j \leq r$}{
set $Q$ to be the last element of $\mathbf{Q}$ and $Y$ to be the last element of $\mathbf{Y}$, \newline
sample $m_i$ vectors uniformly from the $t$-design; construct
$\RR_{Q}$ 
according to 
Def.~\ref{def:postselect}. 

\If{(\ref{eq:golf1}), (\ref{eq:golf2}) hold for $\RR_{Q}$ and $Q \in
T$ with parameters $b=1/8$, $c=1/2$}{
	$\xi_i = 1$ 
	\newline
	$Y \leftarrow  \RR_Q Q - \tr (Q) \Id+ Y$, \quad append $Y$ to $\mathbf{Y}$
	\newline
	$Q \leftarrow X - \PP_T Y$, \quad append $Q$ to $\mathbf{Q}$
}
	$
	i \leftarrow i+1
	$
	}
\eIf{$\sum_{i=1}^l \xi_i = r$ }
	{report \emph{success} and output $\mathbf{Y},\mathbf{Q},\xi$}
	{report \emph{failure}}
\end{algorithm}
\footnotetext{
	The use of pseudo-code allows for a compact presentation of this
	randomized procedure. However, the reader should keep in mind that
	the construction is purely part of a proof and should not be
	confused with the recovery algorithm (which is given in
	Eq.~(\ref{eq:convex_program})).
}

\begin{proof}
The randomzied construction of $Y$
is summarized in
Algorithm~\ref{alg:dual}. 
If this algorithm succeeds, it outputs three lists 
\begin{equation*}
	\mathbf{Y} = 
		\left[ Y_1,\ldots, Y_r 
	\right],\qquad  
	\mathbf{Q} = \left[X, Q_1,\ldots,Q_r \right],
	\qquad \textrm{and}\quad \xi = \left\{\xi_1, \ldots, \xi_l \right\}.
\end{equation*}
The recursive construction 
yields the following
expressions
(c.f.\ \cite[Lemma~14]{kueng_ripless_2013}):
\begin{eqnarray*}
Y &:=& Y_r = \RR_{Q_{r-1}} Q_{r-1} - \tr (Q_{r-1}  ) \Id + Y_{r-1} \\
&=& \sum_{i=1}^r \left( \RR_{Q_{i-1}} Q_{i-1} - \tr (Q_{i-1})\Id   \right) \quad \textrm{and} \\
Q_i &=& X - \PP_T Y_i = \PP_T \left( Q_{i-1} + \tr (Q_{i-1} ) \Id - \RR_{Q_{i-1}} Q_{i-1} \right) \\
&=& \PP_T \left( \II + \Pi_{\Id} - \RR_{Q_{i-1}} \right) Q_{i-1}
= \cdots =  \prod_{j=1}^i \PP_T \left( \II + \Pi_{\Id} - \RR_{Q_{j-1}} \right) X.
\end{eqnarray*}
We now set
\begin{equation}
r = \lceil \log_2 d \rceil + 2. \label{eq:r}
\end{equation}
Then, in case of success, the
validity of properties (\ref{eq:golf1}) and (\ref{eq:golf2}) for
$c=1/2$ and $b=1/8$ in each step
($Q_i \to Q_{i+1}$ and $Y_i \to Y_{i+1}$, respectively) guarantee
\begin{eqnarray*}
\| Y_T - X \|_2 
&=& \| Q_r \|_2
\leq \| X \|_2 \prod_{j=1}^r \frac{1}{2} = 2^{-  \lceil \log_2 d
\rceil - 2} \| X \|_2 \leq \frac{1}{4d}, \\ 
\| Y_T^\perp \|_\infty
&\leq& \sum_{i=1}^r \left\| \PP_T^\perp  \left( \RR_{Q_{i-1}} Q_{i-1} - \tr (Q_{i-1})\Id  \right) \right\|_\infty \\
& \leq&  \sum_{i=1}^r \frac{1}{8} \| Q_{i-1} \|_2 
\leq \frac{1}{8} \sum_{i=1}^r 2^{1-i} \| Q_0 \|_\infty 	\\
&\leq& \| X \|_\infty \frac{1}{8} \sum_{i=0}^\infty 2^{-i} = \frac{1}{4} \leq \frac{1}{2}.
\end{eqnarray*}
Thus, 
$Y_r$ constitutes an approximate dual certificate
in the sense of 
Def.~\ref{def:dual_cert}.  

What remains to be done is to choose the parameters $l$ and $\left\{
m_i \right\}_{i=1}^l$ such that the probability of the
algorithm failing 
is smaller than $0.5 \mathrm{e}^{-\omega}$.
Algorithm \ref{alg:dual} fails precisely if
\begin{equation}
\sum_{i=1}^l \xi_i <  r . \label{eq:dual_aux11} 
\end{equation}
Recall that the $\xi_i$'s are
Bernoulli random variables which indicate whether the $i$-th iteration
of the algorithm has been succesful ($\xi_i = 1$), or failed ($\xi_i = 0$). 
Our aim is to bound the probability of the event in
(\ref{eq:dual_aux11}) by a similar expression involving
\emph{independent}\footnote{
	It was pointed out to us by A.\ Hansen that in some previous papers
	\cite{gross_recovering_2011, kueng_ripless_2013} 
	which involve a similar construction to the one presented here,
	it was
	tacitly assumed that the $\xi_i$ are independent. This will of course
	not be true in general. 
	Fortunately, a more careful argument shows that all
	conclusions remain valid
	\cite{adcock_generalized_2011}. 
	Our treatment here is similar to the one presented in 
	\cite{adcock_generalized_2011}.
}
Bernoulli variables $\xi_i'$. To this end, write
\begin{equation}
\Pr \left[ \sum_{i=1}^l \xi_i <  r \right] 
= \EE \left[ \Pr \Big[ \xi_l < r  - \sum_{i=1}^{l-1} \xi_i\,\Big|\,
\xi_{l-1}, \ldots, \xi_1 \Big] \right]. \label{eq:golfaux1}
\end{equation}
Conditioned on an arbitrary instance of $\xi_{l-1},\ldots,\xi_{1}$, the
variable $\xi_l$ follows a Bernoulli distribution with some parameter
$p(\xi_{l-1}, \dots, \xi_1)$. Note that
if $\xi\sim\operatorname{B}(p)$ is a Bernoulli variable with parameter
$p$, then for every
fixed $t\in\mathbb{R}$, the probability 
$\Pr_{\xi\sim\operatorname{B}(p)}[\xi < t]$ is non-increasing as a
function of $p$.
Consequently, the estimate 
\begin{equation}\label{eqn:ind}
	\Pr \left[ \sum_{i=1}^l \xi_i < r \right] 
	\leq\Pr \left[ \xi_l' + \sum_{i=1}^{l-1} \xi_i <  r \right]  
\end{equation}
is valid if $\xi_l'$ is an independent $p'$-Bernoulli distributed with
\begin{equation*}
	p'\leq \min_{\xi_{l-1}, \dots, \xi_1} p(\xi_{l-1}, \dots, \xi_1). 
\end{equation*}
Proposition~\ref{prop:golfaux} provides a uniform lower bound on the
success
probability $p(\xi_{l-1}, \dots, \xi_1)$. Indeed, 
there is a universal constant $C_1$ such that
invoking Prop.~\ref{prop:golfaux} with
\begin{equation*}
	m := C_1 t d^{2-\gamma} \log d
\end{equation*}
and $Z=Q$
gives a probability of success of at least $9/10$ for any $Q$ (in
particular, independently of the $\xi_{l-1}, \dots, \xi_1$).
Thus, choosing $p'=9/10$ and $m_i = m$ for all $i$, we can
then iterate
the estimate (\ref{eqn:ind}) to arrive at
\begin{equation}
\Pr \left[ \sum_{i=1}^l \xi_i < r \right] 
\leq 
\Pr \left[ \xi_l' + \sum_{i=1}^{l-1} \xi_i < r \right] 
\leq \cdots \leq \Pr \left[ \sum_{i=1}^l \xi_i' < r \right],	\label{eq:golfaux2}
\end{equation}
where the $\xi_i'$ are independent Bernoulli variables with parameter
$9/10$.
A standard Chernoff bound
(e.g.\ \cite[Section Concentration:
Theorem~2.1]{habib_probabilistic_1998})
gives 
\begin{equation*}
\Pr \left[ \sum_{i=1}^l \xi_i ' \leq l(9/10 -t) \right] \leq \mathrm{e}^{-2lt^2}.
\end{equation*}
Choosing $t = 9/10 - r/l$ we obtain
\begin{eqnarray}
\Pr \left[ \sum_{i=1}^l \xi_i' < r \right]
& \leq & \Pr \left[ \sum_{i=1}^l \xi_i' \leq r \right]
= \Pr \left[ \sum_{i=1}^l \xi_i' \leq l \left( 9/10 - t \right) \right] \nonumber \\
&\leq & \exp \left( - 2l \left( \frac{9}{10} - \frac{r}{l} \right)^2 \right). \label{eq:dual_aux12}
\end{eqnarray}
Setting the number of iterations generously to
\begin{equation*}
l = 10 \omega r = 10 \omega \left( \lceil \log_2 d \rceil + 2 \right)
\end{equation*}
implies
\begin{equation*}
2l \left( \frac{9}{10} - \frac{r}{l} \right)^2
\geq 20 \omega r \left( \frac{8}{10} \right)^2 \geq 12 \omega r \geq \omega + \log 2,
\end{equation*}
where we have used $\omega \geq 1 \geq \log 2$ in the last inequality. 
Together with (\ref{eq:dual_aux11}), (\ref{eq:golfaux2}) and (\ref{eq:dual_aux12}) this gives the desired bound
\begin{eqnarray*}
\Pr \left[ \textrm{algorithm fails} \right]
&=& \Pr \left[ \sum_{i=1}^l \xi_i < r \right]
\leq \Pr \left[ \sum_{i=1}^l \xi_i' < r \right]
 \leq  \mathrm{e}^{- \omega - \log (2) } = \frac{1}{2} \mathrm{e}^{-\omega},
\end{eqnarray*}
on our construction of $Y$ failing. 
The total number of measurement vectors sampled is
\begin{equation*}
	\sum_{i=1}^l m_i
	= l m_l
  \leq  C \omega t d^{2-\gamma} \log^2 d,
\end{equation*}
for some constant $C$.
\end{proof}

Finally we are ready to put all pieces together and show or main result -- Theorem \ref{thm:main_theorem}.

\begin{proof}[Proof of the Main Theorem]
In section \ref{sec:convex_geometry} (Proposition \ref{prop:convex geometry}) we have shown that the algorithm (\ref{eq:convex_program})
recovers the sought for signal $x$, provided that (\ref{eq:inj1}) holds and a suitable approximate dual certificate $Y$ exists. 
Proposition \ref{prop:Y} -- with a maximal truncation rate of $\gamma = (1-2/t)$ -- implies that the probability that no such $Y$ can be constructed
is smaller than $0.5 \mathrm{e}^{-\omega}$, provided that the sampling rate $m$ obeys
\begin{equation}
m  \geq C \omega t d^{1+ 2/t} \log^2 d,	\label{eq:sampling_rate}
\end{equation}
for a sufficiently large absolute constant $C$. Provided that this
constant is large enough,
Proposition~\ref{prop:inj1}
implies that the probability of (\ref{eq:inj1}) failing is also bounded by $0.5 \mathrm{e}^{-\omega}$. 
Theorem \ref{thm:main_theorem} now follows from the union bound over these two probabilities of failure. 
\end{proof}

\section{Converse Bound}	\label{sec:converse_bounds}

In this paper, we require designs of order at least three. Here we prove that this criterion is fundamental in the sense that sampling from 2-designs in general cannot guarantee a sub-quadtratic sampling rate.
In order to do so, we will use a particular sort of 2-design, called a \emph{maximal set of mutually unbiased bases} (MUBs) \cite{schwinger_unitary_1960, zauner_quantendesigns_1999, konig_cubature_1999,
klappenecker_mutually_2005}. 
Two orthonormal bases $\left\{u_i \right\}_{i=1}^d$ and $\left\{v_i
\right\}_{i=1}^d$ are called \emph{mutually unbiased} if their overlap
is uniformly minimal. Concretely, this means that
\begin{equation*}
| \langle u_i, v_j \rangle |^2 = \frac{1}{d} \quad \forall i,j=1, \ldots, d
\end{equation*}
must hold for all $i,j=1,\ldots, d$. Note that this is just a generalization of the incoherence property between standard and Fourier basis. 
In prime power dimensions, a maximal set of $(d+1)$ such MUBs is known to exist (and can be constructed) \cite{klappenecker_constructions_2004}. 
Such a set is maximal in the sense that it is not possible to find more than $(d+1)$ MUBs in any Hilbert space.
Among other interesting properties -- cf. \cite{durt_mutually_2010}
for a detailed survey  --  maximal sets of MUBs are known to form
2-designs \cite{zauner_quantendesigns_1999, klappenecker_mutually_2005}.

The defining properties of a maximal set of  MUBs allow us to derive the converse bound -- Theorem \ref{thm:converse_bound}. 

\begin{theorem}[Converse bound] \label{thm:converse_bound1}
	Let $d \geq 2$ be a prime power and let $D_2\subset
	\CC^d$ be a maximal set of MUBs. 
	Then there exist orthogonal, normalized vectors $x, z\in\CC^d$ 
	which have the following property.
	
	Suppose that 
	$m$ measurement vectors $a_1,
	\dots, a_m$ are sampled
	independently and
	uniformly at random from $D_2$.
	Then, for any $\omega \geq 0$, the number of measurements must obey
	\begin{equation}
		m \geq \frac{\omega}{4}d(d+1), \label{eq:converse_bound1}
	\end{equation}
	or the event 
	\begin{equation*}
		|\langle a_i, x\rangle|^2 =
		|\langle a_i, z\rangle|^2
		\quad
		\forall\, i\in\{1, \dots, m\}
	\end{equation*}
	will occur
	with probability at
	least $\mathrm{e}^{-\omega}$. 
\end{theorem}

Consequently a scaling of $\mathcal{O}(d^2)$ in general cannot be avoided when demanding only the property of being a 2-design and simultaneously requiring a ``reasonably small'' probability of failure
in the recovery process.

\begin{proof}[Proof of Theorem \ref{thm:converse_bound1}]

	Suppose that $\left\{ u_i \right\}_{i=1}^d$ is one orthonormal basis contained in the maximal set of MUBs $D_2$ 
	and set $x:=u_1$ as well as $z := u_2$. Note that by definition these vectors are orthogonal and normalized. 
	Due to the particular structure of MUBs,
	$x$ and $z$ can only be distinguished if either $u_1$ or $u_2$
	is contained in $\left\{a_1, \ldots, a_m \right\}$.  Since each
	$a_i$ is chosen iid at random from $D_2$ containing $(d+1)d$
	elements, the probability of obtaining either $u_1$ or $u_2$ is
	$p=\frac{2}{(d+1)d}$.  
	As a result, the problem reduces to the following standard stopping
	time problem (cf.\ for example	Example (2) in Chapter 6.2 in
	\cite{feller_introduction_2008}):

	Suppose that the probability of success in a Bernoulli experiment is
	$p$. How many trials $m$ are required in order for the probability
	of at least one success to be $1-\mathrm{e}^{\omega}$ or larger?

	To answer this question, we have to find the smallest integer $m$
	such that 
	\begin{equation} 
		1-(1-p)^m \geq 1- \mathrm{e}^{-\omega},
		\quad \textrm{or equivalently} \quad - m \log (1-p) \geq \omega.
		\label{eq:conaux1} 
	\end{equation} 
	The standard inequality
	\begin{equation*}
	p \leq - \log (1-p) \leq \frac{p}{1-p} \leq 2p
	\end{equation*}
	for any $p \in [0,1/2]$ implies that (\ref{eq:converse_bound1}) is a necessary criterion for (\ref{eq:conaux1})
	and we are done.
\end{proof}

\section{Conclusion}

In this paper we have derived a partly derandomized version of Gaussian PhaseLift \cite{candes_phaselift_2012, candes_solving_2012}.
Instead of Gaussian random measurements, our method guarantees recovery for sampling iid from certain finite vector configurations, dubbed $t$-designs.
The required sampling rate depends on the design order $t$:
\begin{equation}
m = \mathcal{O} \left( t d^{1+2/t} \log^2 d \right).	\label{eq:discussion1}
\end{equation}
For small $t$ this rate is worse than the Gaussian analogue -- but still non-trivial. However, as soon as $t$ exceeds $2 \log d$, we obtain linear scaling up to a polylogarithmic overhead.  

In any case, we feel that the main purpose of this paper is not to
present yet another efficient solution heuristics, but to show that
the phase retrieval problem can be derandomized using $t$-designs.
These finite vector sets lie in the vast intermediate region between
random Fourier vectors and Gaussian random vectors (the Fourier basis
is a $1$-design, whereas normalized Gaussian random vectors correspond to an
$\infty$-design). Therefore the design order $t$ allows us to 
gradually transcend between these two extremal cases.

\textbf{Acknowledgements}
DG and RK are grateful to the organizers and participants of the
Workshop on Phaseless Reconstruction, held as part of the 2013
February Fourier Talks at the University of Maryland, where they were
introduced to the details of the problem. This extends, in particular,
to Thomas Strohmer.

The work of DG and RK is supported by the Excellence Initiative of the
German Federal and State Governments (Grant ZUK 43), by scholarship
funds from the State Graduate Funding Program of Baden-W\"urttemberg,
and by the US Army Research Office under contracts W911NF-14-1-0098
and W911NF-14-1-0133 (Quantum Characterization, Verification, and
Validation), and the DFG. FK acknowledges support from the German Federal Ministry
of Education and Reseach (BMBF) through the cooperative research
project ZeMat.

\bibliographystyle{IEEEtran}
\bibliography{phaseless_bib}

\section{Appendix}

Here we briefly state an elementary proof of Lemma \ref{lem:tr1sym2}.
In the main text we proved this result using wiring diagrams.  The
purpose of this is to underline the relative simplicity of wiring diagram
calculations. Indeed, the elementary proof below is considerably more
cumbersome than its pictorial counterpart. 

\subsection{Elementary proof of Lemma \ref{lem:tr1sym2}} \label{sub:alternative_proof}

Let us choose an arbitrary orthonormal basis $b_1, \ldots ,b_d$ of $V^d$. 
In the induced basis $\left\{ b_i \otimes b_j \right\}_{i,j=1}^d$ of
$V^d \otimes V^d$ the transpositions then correspond to
\begin{equation*}
\underline{1} = \Id \otimes \Id = \sum_{i=1}^d b_i b_i^* \otimes \sum_{j=1}^d b_j b_j^*
\quad \textrm{and} \quad \sigma_{(1,2)} = \sum_{i,j=1}^d b_i b_j^* \otimes b_j b_i^*.
\end{equation*}
This choice of basis furthermore allows us to write down $\tr_2 (A)$ for $A \in M^d \otimes M^d$ explicity:
\begin{equation*}
\tr_2 (A) = \sum_{i=1}^d \left( \Id \otimes b_i^* \right) A \left( \Id \otimes b_i \right).
\end{equation*}
Consequently we get for $A,B \in H^d$ arbitrary
\begin{equation*}
\tr_2 \left( P_{\Sym^2} A \otimes B \right)
= \frac{1}{2} \tr_2 \left( A \otimes B \right) + \frac{1}{2} \tr_2 \left( \sigma_{(1,2)} A \otimes B \right).
\end{equation*}
The latter term can be evaluated explicitly:
\begin{eqnarray*}
\tr_2 \left( \sigma_{(1,2)} A \otimes B \right) 
&=& \sum_{k=1}^d \left( \Id \otimes b_k^* \right) \sum_{i,j=1}^d  b_i b_j^* \otimes b_j b_i^* A \otimes B \left( \Id \otimes b_k \right) \\
&=& \sum_{i,j,k=1}^d b_i b_j^* A b_k^* b_j b_i^* B b_k
= \sum_{i,j=1}^d \langle b_i, B b_j \rangle b_i b_j^* A 	\\
&=& \left( \sum_{i=1}^d b_i b_i ^* \right) B \left( \sum_{j=1}^d b_j b_j^* \right)A 
= \Id B \Id A
= B A,
\end{eqnarray*}
and the desired result follows. Here we have used the basis representation of the identity, namely $\Id = \sum_{i=1}^d b_i b_i^*$.

\end{document}

%% file: tikzstyledefs.tex
\tikzstyle heightone=[scale=.7,shift={(0,-.3)}]
\tikzstyle heightones=[scale=.8,xscale=.35,shift={(0,.1)}]
\tikzstyle heightoneonehalf=[scale=.9,shift={(0,-.2)}]
\tikzstyle heighttwo=[scale=.9,shift={(0,-.4)}]
\tikzstyle heighttwos=[scale=.5,xscale=.6,shift={(0,-.1)}]
\tikzstyle heightthree=[scale=.6,shift={(0,-.9)}]
\tikzstyle heightthrees=[scale=.4,xscale=.7,shift={(0,-.2)}]

\tikzstyle arrowstyle=[blue,semitransparent,scale=2]

\tikzstyle basiclabel=[draw=none,fill=none,shape=rectangle,inner sep=2pt,scale=.8]
\tikzstyle leftlabel=[basiclabel,anchor=east]
\tikzstyle rightlabel=[basiclabel,anchor=west]
\tikzstyle bottomlabel=[basiclabel,anchor=north]
\tikzstyle toplabel=[basiclabel,anchor=south]

\tikzstyle vertex=[circle,draw,fill=black,inner sep=1pt]
\tikzstyle ciliation=[circle,draw=none,fill=red,inner sep=1pt,semitransparent]
\tikzstyle ciliatednode=[vertex,pin={[pin distance=1mm,pin edge={semitransparent,red},ciliation]#1:{}}]

\tikzstyle vector=[black,thick,rectangle,draw=gray!50!yellow,top color=yellow!30,bottom color=black!10,scale=.8,inner sep=2pt]
\tikzstyle small vector=[vector,scale=.8]
\tikzstyle plain vector=[rectangle,draw=none,fill=white,scale=.7]

\tikzstyle my signal=[black,thick,signal,signal pointer angle=120,draw=blue!50,top color=blue!20,bottom color=black!10,scale=.8,inner sep=2pt]
\tikzstyle matrix=[my signal,signal from=south,signal to=north]
\tikzstyle reverse matrix=[my signal,signal from=north,signal to=south]

\tikzstyle small matrix=[matrix,scale=.7]
\tikzstyle reverse small matrix=[reverse matrix,scale=.7]
\tikzstyle matrix on edge=[small matrix,sloped,rotate=-90]
\tikzstyle reverse matrix on edge=[small matrix,sloped,rotate=90]

\tikzstyle trivalent=[very thick]
\tikzstyle dotdotdot=[decorate,decoration={markings,
    mark=at position .3 with{\node{.};},
    mark=at position .5 with {\node{.};},
    mark=at position .7 with {\node{.};}}]

\tikzstyle wavyup=[out=90,in=-90]
\tikzstyle wavydown=[out=-90,in=90]

\tikzstyle symmetrizer=[rectangle,fill=gray!10,draw=black]
\tikzstyle permutation=[symmetrizer]
\tikzstyle antisymmetrizer=[rectangle,fill=black,draw=black]
\tikzstyle symlabel=[draw=none,fill=none,black,scale=.8]
\tikzstyle asymlabel=[draw=none,fill=none,white,scale=.8]